\documentclass[twoside,11pt]{article}

\usepackage{jmlr2e,amsmath}
\newtheorem{property}{Property}
\begin{document}

\title{Message Error Analysis of Loopy Belief Propagation for the Sum-Product Algorithm}

\author{\name Xiangqiong Shi \email xshi4@uic.edu \\
       \name Dan Schonfeld \email dans@uic.edu \\
       \name Daniela Tuninetti \email danielat@uic.edu \\
       \addr Department of Electrical and Computer Engineering\\
       University of Illinois at Chicago\\
       Chicago, IL , USA}

\editor{}

\maketitle

\begin{abstract}%   <- trailing '%' for backward compatibility of .sty file
Belief propagation is known to perform extremely well in many
practical statistical inference and learning problems using
graphical models, even in the presence of multiple loops. The
iterative use of belief propagation algorithm on loopy graphs is
referred to as Loopy Belief Propagation (LBP). Various sufficient
conditions for convergence of LBP have been presented; however,
general necessary conditions for its convergence to a unique fixed
point remain unknown. Because the approximation of beliefs to true
marginal probabilities has been shown to relate to the convergence
of LBP, several methods have been explored whose aim is to obtain
distance bounds on beliefs when LBP fails to converge. In this
paper, we derive uniform and non-uniform error bounds on messages,
which are tighter than existing ones in literature, and use these
bounds to derive sufficient conditions for the convergence of LBP
in terms of the sum-product algorithm. We subsequently use these
bounds to study the dynamic behavior of the sum-product algorithm,
and analyze the relation between convergence of LBP and sparsity
and walk-summability of graphical models. We finally use the
bounds derived to investigate the accuracy of LBP, as well as the
scheduling priority in asynchronous LBP.
\end{abstract}


\section{Introduction}
\label{sec:intro}

Probabilistic inference for large-scale multivariate random
variables is very expensive computationally.  Belief propagation
(BP) algorithms are designed to reduce the computational burden by
exploiting the factorization of joint density functions captured by
the topological structure of graphical
models~[\cite{Bishop_06,Jordan_99,Kschischang_01,Wainwright_08}]. BP
is known to converge to the exact inference on acyclic graphs (i.e.
trees) or graphs that contain a single loop. In the case of graphs
with multiple loops, BP results in an iterative method referred to
as loopy belief propagation (LBP). The use of LBP generally provides
remarkably good approximations in real-world applications; e.g.,
turbo decoding and stereo
matching~[\cite{Mceliece_98,Sun03stereomatching}].

Because LBP does not always converge, sufficient conditions for
its convergence have been extensively investigated in the past
using various
approaches~[\cite{TatikondaJ02,heskes2004a,ihler05b,MooijKappen_IEEETIT_07}].
Necessary conditions for convergence of LBP, however, remain
unknown. ~\cite{TatikondaJ02} related convergence of LBP to the
uniqueness of a sequence of Gibbs measures defined on the
associated computation tree. He subsequently developed a testable
sufficient condition for convergence of LBP by applying Simon's
condition~[\cite{Georgii_1988}]. ~\cite{heskes2004a} presented
sufficient conditions for uniqueness of fixed points in LBP by
relying on the uniqueness of minima of the Bethe free energy. He
related the strength of the potentials with the convergence of the
LBP algorithm, which leads to better sufficient conditions than
those exclusively relying on the structure of the graph.

Recently, several papers have investigated the message updating
functions of the LBP algorithm as contractive mappings.
~\cite{ihler05b} analyzed the contractive dynamics of
message-error propagation in belief networks using dynamic-range
measure as a metric, and obtained error bounds and sufficient
conditions for convergence of LBP message passing.
~\cite{MooijKappen_IEEETIT_07} derived sufficient conditions for
convergence of LBP based on quotient norms of contractive
mappings, which are invariant to scaling and shown to be valid for
potential functions containing zeros.

For Gaussian graphical models, \cite{Malioutov2006} related the
convergence of means and variances to walk sums and defined
walk-summability with respect to spectral radius of partial
correlation coefficient matrix. For binary graphs,
\cite{NIPS2009_Yusuke} presented an edge zeta function based on
weighted prime cycles, and related convexity of Bethe free energy
with the determinant formula of edge zeta function. They showed
similar walk-summability of binary graphs by relating the spectra
of correlation coefficient matrix with Hessian of Bethe free
energy. For general graphical models,
\cite{MooijKappen_IEEETIT_07} derived certain interaction
coefficients between random variables based on strength of
potential functions, and related the spectral radius of
coefficient matrix with the convergence of LBP. Enlightened by
those similar analysis, we defined walk-summable for general
graphs and compared walk-summability with other existing
convergence conditions.

Although the beliefs may not be true marginal probabilities when the
LBP algorithm converges, they have been shown to provide good
approximations by~\cite{Weiss_00}. When the LBP algorithm does not
converge, however, beliefs are not good approximations of true
marginals because the Bethe free energy does not provide a good
approximation of the Gibbs-Helmholtz free energy~[\cite{Yedidia04}].
Exactness and accuracy of the LBP algorithm has consequently gained
interest in recent years. ~\cite{Tatikonda03} derived bounds on
exact marginals by relying on the girth of the graph (i.e. the
number of edges in the shortest cycle in the graph) and the
properties of Dobrushin's interdependence matrix~[\cite{Salas1997}].
~\cite{TagaM06MICAI} used Dobrushin's theorem to present a distance
bound on the marginal probabilities. ~\cite{ihler07b} introduced a
distance bound on the error between beliefs and marginals based on
recent results for computing marginal probabilities for pairwise
Markov random fields using Self-Avoiding Walk (SAW)
trees~[\cite{Weitz06}]. ~\cite{MooijKappen_NIPS_08} propagate bounds
on marginal probabilities over a subtree or the SAW tree of the
factor graph, and demonstrate that their bounds perform well in
terms of accuracy and computation time of LBP.

Several investigators have explored the consequence of scheduling on
the convergence of BP. ~\cite{TagaM06IEICE} discussed the impatient
and lazy belief propagation algorithms and showed that the former is
expected to converge faster than the latter. ~\cite{Elidan06}
proposed a residual belief propagation algorithm, which schedules
messages in an informed manner thus significantly reducing the
running time needed for convergence of LBP. Inspired by
~\cite{Elidan06}'s work, ~\cite{Sutton07} further increased the rate
of convergence by estimating the residual rather than computing it
directly.

In this paper, we derive tight error bounds on LBP and use these
bounds to study the dynamics---error, convergence, accuracy, and
scheduling---of the sum-product algorithm.\footnote{A preliminary
version of some of the error bounds presented in this paper has
appeared in~\cite{Shi2009}.} Specifically, in
Section~\ref{sec:message error} and Section~\ref{sec:bounds
belief}, we rely on the contractive mapping property of message
errors to present novel uniform and non-uniform distance bounds
between multiple fixed-point solutions. Several graphical networks
are investigated and used to demonstrate that the proposed
distance bounds are tighter than existing bounds. We subsequently
use these bounds to derive uniform and non-uniform sufficient
conditions for convergence of the sum-product algorithm. Moreover,
in Section~\ref{sec:LBP convergence}, we analyze the relation
between convergence and sparsity of graphs, and extend the
convergence perspective of walk-summability from Gaussian
graphical models to general graphical models. In
Section~\ref{sec:accuracy}, we present bounds on the distance
between beliefs and true marginals by applying SAW trees and show
that the proposed bounds can be used to improve existing bounds.
Furthermore, in Section \ref{sec:residual scheduling}, we explore
the use of the upper-bound on message errors as a criterion to
rank the priority of message passing for scheduling in
asynchronous LBP. We then present a case study of LBP by studying
its dynamics on completely uniform graphs and analyzing its true
fixed points and message-error functions in
Section~\ref{sec:uniform graph}. We conclude the paper in
Section~\ref{sec:conclusion}.

\section{Message-Error Propagation for the Sum-Product Algorithm}
\label{sec:message error}

Belief propagation originated from exact inference on tree
structured graphical models, though for graphs with loops it shows
remarkable performance of approximate inference. BP is
synonymously called sum-product algorithm for marginalization of
global distribution or max-product algorithm to compute
Maximum-A-Posteriori (MAP). In this paper, we will mainly talk
about sum-product algorithm for graphs with loops.

\subsection{Loopy Belief Propagation Updates}
\begin{figure}[htb]
\begin{center}
\includegraphics[height=4.64 cm,width=12.79 cm]{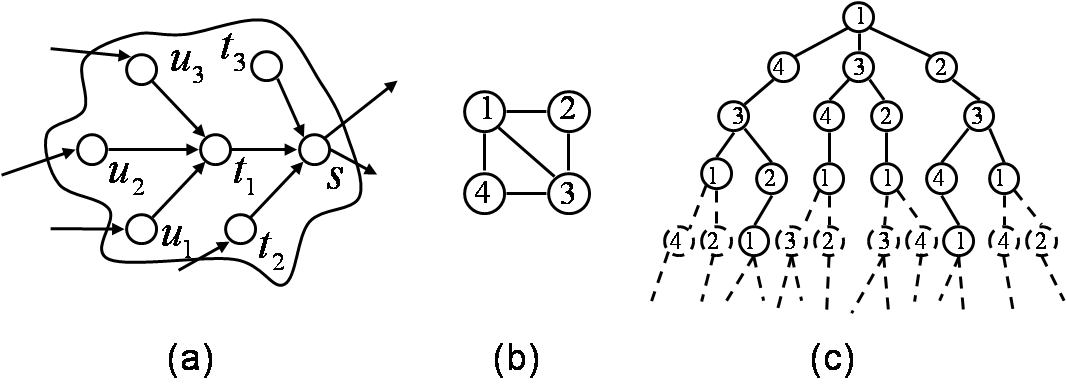}
\end{center}
\caption{Graphical models: (a) message passing in a portion of a
belief network; (b) a simple graph; and (c) Bethe tree (all nodes
and edges) and Self-Avoiding Walk tree (black solid only) of (b).}
\label{fig:msg passing}
\end{figure}
Let us consider a general graphical model
$\mathbb{G}=(\mathbb{V},\mathbb{E})$ whose distribution factors as
follows:
\begin{equation}
p(X)=\frac{1}{Z}\prod_{(s,t)\in
\mathbb{E}}\psi_{st}(x_s,x_t)\prod_{s\in\mathbb{V}} \psi_s(x_s),
\label{eq:pdf}
\end{equation}
where $Z$ is a normalization factor, $\psi_{st}(x_s,x_t)$ is the
pairwise potential function between random variables $x_s$ and
$x_t$, and $\psi_s(x_s)$ is the single node potential function on
$x_s$. $(s,t)$ denotes an undirected edge, $\mathbb{V}$ is the set
of nodes, and $\mathbb{E}$ is the set of edges. We assume that all
the potential functions are positive.
%Assuming continuous random variables,

Fig.~\ref{fig:msg passing}(a) illustrates the message passing
mechanism used in BP. The updating rule of the sum-product algorithm
for the message sent by node $t$ to its neighbor node $s$ at
iteration $i$ is:
\begin{equation}
m_{ts}^i(x_s)\propto\int\psi_{ts}(x_t,x_s)\psi_t(x_t)\prod_{u\in
\Gamma_t\backslash s}m_{ut}^{i-1}(x_t)dx_t, \label{eq:msg}
\end{equation}
where $\Gamma_t$ is the set of neighbors of node $t$. The belief,
or pseudo-marginal probability of $x_t$, on node $t$ at iteration
$i$, is:
\begin{equation}
B_t^i(x_t)\propto\psi_t(x_t)\prod_{u\in
\Gamma_t}m_{ut}^i(x_t).\label{eq:singlebelive}
\end{equation}
A stable fixed point has been reached if
$m_{ts}^i(x_s)=m_{ts}^{i+1}(x_s)$, $\forall s\in\mathbb{V}$. The
pairwise belief of random variables $x_s,x_t$ at iteration $i$ is
defined as:
\begin{equation}
B_{ts}^i(x_t,x_s)\propto\psi_{ts}(x_t,x_s)\psi_t(x_t)\psi_s(x_s)\prod_{u\in
\Gamma_t\backslash s}m_{ut}^i(x_t)\prod_{p\in \Gamma_s\backslash
t}m_{ps}^i(x_s).\label{eq:pairbelive}
\end{equation}

The computation tree first introduced in~\cite{Wiberg_1996} is
always applied in the analysis of LBP.  Bethe tree and SAW tree
are two types of computation trees used in~\cite{ihler07b}, which
will also be used in the rest of the paper. Both Bethe tree and
SAW tree are tree-structured unwrappings of a graph $\mathbb{G}$
from some node $v$. The Bethe tree, denoted as
$T_B(\mathbb{G},v,n)$, contains all paths of length $n$ from $v$
that do not backtrack, while the SAW tree, denoted as
$T_{SAW}(\mathbb{G},v,n)$, contains all paths of length $n\leq
|\mathbb{V}|+1$ that do not backtrack and have all nodes on the
path unique. The belief on node $v$ at iteration $n$ in
synchronous LBP is equivalent to the exact marginal of the root
$v$ in the $n$-level Bethe tree.

Figure~\ref{fig:msg passing}(c) illustrates the Bethe tree and the
SAW tree for the graphical model in Figure~\ref{fig:msg
passing}(b). For synchronous BP, each iteration of
Equations~\eqref{eq:msg},~\eqref{eq:singlebelive}
and~\eqref{eq:pairbelive} corresponds to a level in the Bethe
tree.

\subsection{Approaches to Analyze Convergence of LBP}
Various approaches have been presented to derive convergence
conditions for the sum-product algorithm, including Gibbs measure
[\cite{TatikondaJ02}], equivalent minimax problem
[\cite{heskes2004a}], and contraction property of LBP updates
[\cite{ihler05b,MooijKappen_IEEETIT_07}]. \cite{TatikondaJ02}
proved that, when the Gibbs measure on the corresponding
computation tree is unique, LBP converges to a unique fixed point.
\cite{heskes2004a} proved that, when the minima of Bethe free
energy is unique, there is a unique fixed point for LBP.
\cite{ihler05b} and \cite{MooijKappen_IEEETIT_07} used similar
methodology by applying $\ell_\infty$ measure on potential
functions. They proved that when LBP updating is a contractive
mapping, LBP will converge. They both compared their convergence
results with those of \cite{TatikondaJ02} and \cite{heskes2004a},
and showed that their results are stronger.
\cite{MooijKappen_IEEETIT_07} further showed that they derived
more general results than \cite{ihler05b}. Enlightened by the
discussion in \cite{ihler05b} and \cite{MooijKappen_IEEETIT_07},
and based on the framework of \cite{ihler05b}, we use a new
measure on message errors of LBP, in order to obtain distance
bound and accuracy bound.

Our contributions are as follows:

1. We present a tight upper- and lower- bound for multiplicative
message error $e(x)$ in Section \ref{sec:upper lower bd}.
Furthermore, based on the upper- and lower- bound, we derive tight
uniform distance bound and non-uniform distance bound for beliefs
$B(x)$ in Section \ref{sec:bounds belief}, which help to tighten
the accuracy bounds between beliefs and true marginals in Section
\ref{sec:accuracy} and correct the upper-bound on message
residuals for residual scheduling in Section \ref{sec:residual
scheduling}.

2. We investigate the relation between convergence of LBP with
sparsity and walk-summability of graphical models in Section
\ref{sec:LBP convergence}. We extend walk-summability for Gaussian
graphical models to general graphical models and compare the
tightness of existing convergence conditions.

3. We analyze the paramagnetic fixed point, ferromagnetic and
antiferromagnetic fixed points for uniform binary graphs using
message updating functions, and present true message error
variation functions to show dynamics of sum-product algorithm in
Section \ref{sec:uniform graph}.

\subsection{Message-Error Measures}
Define {\em message error} as a multiplicative function
$e_{ts}^i(x_s)$ that perturbs the fixed-point message
$m_{ts}(x_s)$. The perturbed message at iteration $i$ is hence
\[\hat{m}_{ts}^i(x_s)=m_{ts}(x_s)e_{ts}^i(x_s).\]Dealing with
normalized messages, we define {\em fixed-point incoming message
products} as
\[M_{ts}(x_t)\propto
\psi_t(x_t)\prod_{u\in\Gamma_t\backslash s}m_{ut}(x_t),\] and {\em
perturbed incoming message products} as
\[M_{ts}^i(x_t)\propto
\psi_t(x_t)\prod_{u\in\Gamma_t\backslash s}m_{ut}^i(x_t),\] and
{\em incoming error products} as
\[E_{ts}^i(x_t)=\prod_{u\in\Gamma_t\backslash
s}e_{ut}^i(x_t).\] We have \[M_{ts}^i(x_t)\propto
M_{ts}(x_t)E_{ts}^i(x_t).\] Thus, the {\em outgoing message error}
from node $t$ to node $s$ at iteration $i+1$ is:
\begin{align*}
e_{ts}^{i+1}(x_{s})=\frac{\hat{m}_{ts}^{i+1}(x_s)}{m_{ts}(x_s)}
=\frac{\int\psi_{ts}(x_t,x_s)M_{ts}(x_t) E_{ts}^i(x_t)
dx_t}{\int\psi_{ts}(x_t,x_s)M_{ts}(x_t) E_{ts}^i(x_t) dx_t dx_s}
\times\frac{\int\psi_{ts}(x_t,x_s)M_{ts}(x_t)dx_t
dx_s}{\int\psi_{ts}(x_t,x_s)M_{ts}(x_t)dx_t}.
\end{align*}

In the following, we will introduce two measures on message
errors.

\subsubsection{Dynamic-Range Measure}
\label{sec:dynamic range measure}

The {\em dynamic-range measure} of error introduced
by~\cite{ihler05b} is defined as:
\begin{equation}
d(e^i_{ts})=\max_{a,b}\sqrt{\frac{e^i_{ts}(a)}{e^i_{ts}(b)}}.
\label{eq:dynamic range}
\end{equation}
We have $d(e^i_{ts})\to 1$ when $e^i_{ts}(x)\to 1$. In
~\cite{ihler05b}~[Th.8] it was shown that when
$d(\psi_{ts})=\max_{a,b,c,d}\sqrt{\frac{\psi_{ts}(a,b)}{\psi_{ts}(c,d)}}$
is finite, the dynamic-range measure satisfies the following
contraction:
\begin{equation}
d(e_{ts}^{i+1})\leq\frac{d(\psi_{ts})^2d(E_{ts}^i)+1}{d(\psi_{ts})^2+d(E_{ts}^i)},
\label{eq:dynamiccontraction}
\end{equation}
in other words, based on the dynamic-range measure, the outgoing
message error is bounded by a non-linear function of the potential
function and the incoming error product.

\subsubsection{Maximum-Error Measure}
To study the dynamics of message error propagation, dealing
directly with errors is more interesting than dealing with dynamic
range. Moreover, we target to tighten distance bounds of LBP
results by using a new error measure. We thus introduce the
following {\em maximum multiplicative error} function as an error
measure:
\begin{equation}
\max_{x_s} e^{i+1}_{ts}(x_{s})=\max_{x_s}
   \frac{\int\psi_{ts}(x_t,x_s)M_{ts}(x_t)E^i_{ts}(x_t)dx_t}{\int\psi_{t\star}(x_t)M_{ts}(x_t)E^i_{ts}(x_t)dx_t}\times
\frac{\int\psi_{t\star}(x_t)M_{ts}(x_t)dx_t}{\int\psi_{ts}(x_t,x_s)M_{ts}(x_t)dx_t},\label{eq:max
measure}
\end{equation}
where $\psi_{t\star}(x_t) = \int\psi_{ts}(x_t,x_s)dx_s$. It is
immediate that the maximum-error measure approaches one when
multiplicative errors vanish. We will show later that this error
measure satisfies the following contraction:
\begin{equation}
\max_{x_s}e_{ts}^{i+1}(x_s)
\leq\left(\frac{d(\psi_{ts})d(\psi_{t\star})d(E_{ts}^i)+1}{d(\psi_{ts})d(\psi_{t\star})+d(E_{ts}^i)}\right)^2.\label{eq:maxcontraction}
\end{equation}

Dynamic-range measure and maximum-error measure are equivalent
when the maximum and minimum of an error function are reciprocal.
By comparison, maximum-error measure gives an absolute error,
while dynamic-range measure gives a relative error which is
invariant to scaling. We will show in the following of the paper
that maximum-error measure should be used, when we are interested
in absolute errors. Furthermore, both defined in dynamic-range
measure, $d(\psi_{ts})$ and $d(\psi_{t\star})$ correspond to two
types of matrix norms on $\psi_{ts}$. $d(\psi_{t\star})$ in the
RHS of Inequality~\eqref{eq:maxcontraction} characterizes the
effect of normalization factor on $\max_{x_s}
e^{i+1}_{ts}(x_{s})$. We will discuss the influence of
$d(\psi_{t\star})$ on error bounds in Section~\ref{sec:upper lower
bd}.

\subsection{Strength of Potential Functions}
\cite{heskes2004a},~\cite{ihler05b}
and~\cite{MooijKappen_IEEETIT_07} have defined measures of
strength of potential functions respectively, which help to obtain
better convergence conditions than those only related with
topology of graphical models. In the following, we will show the
relationship between beliefs and strength of pairwise potential
functions.

\subsubsection{Strength of Potential functions in~\cite{heskes2004a}}

\cite{heskes2004a} defined $\sigma_{t,s}$ as the strength of a
pairwise potential function $\psi_{ts}(x_t,x_s)$ meeting the
following equation:
\begin{equation*}
\frac{1}{1-\sigma_{t,s}}=\max_{x_t,x_s,\hat{x}_t,\hat{x}_s}\frac{\psi_{ts}(x_t,x_s)\psi_{ts}(\hat{x}_t,\hat{x}_s)}{\psi_{ts}(x_t,\hat{x}_s)\psi_{ts}(\hat{x}_t,x_s)}.
\end{equation*}
This strength is related with the correlation of LBP marginals as
follows:
\begin{equation*}
\frac{B_{ts}(x_t,\hat{x}_s)}{B_t(x_t)B_s(\hat{x}_s)}\leq\frac{1}{1-\sigma_{t,s}},
\end{equation*}
which was then utilized to give a better convergence condition
than the one only depending on graph topology.

\subsubsection{Strength of Potential functions in~\cite{ihler05b}}

\cite{ihler05b} proposed the dynamic-range measure $d(\psi_{ts})$
as the strength of potential functions $\psi_{ts}(x_t,x_s)$. Let
us restate the definition of the strength of potential functions
and its relationship with message errors in Section
\ref{sec:dynamic range measure} as follows:
\begin{eqnarray*}
&d(\psi_{ts})=\max_{x_t,x_s,\hat{x}_t,\hat{x}_s}\sqrt{\frac{\psi_{ts}(x_t,x_s)}{\psi_{ts}(\hat{x}_t,\hat{x}_s)}},\\
&d(e_{ts})\leq\frac{d(\psi_{ts})^2d(E_{ts})+1}{d(\psi_{ts})^2+d(E_{ts})}.
\end{eqnarray*}
By considering single node potentials $\psi_{t}(x_t)$ and
$\psi_{s}(x_s)$, \cite{ihler05b} weakened the strength of pairwise
potential functions by using the following dynamic range measure:
\begin{equation}
d(\psi_{ts})^2=\min_{\psi_t,\psi_s}d(\frac{\psi_{ts}}{\psi_t\psi_s})^2=\sup_{x_t,x_s,\hat{x}_t,\hat{x}_s}\sqrt{\frac{\psi_{ts}(x_t,x_s)\psi_{ts}(\hat{x}_t,\hat{x}_s)}{\psi_{ts}(\hat{x}_t,x_s)\psi_{ts}(x_t,\hat{x}_s)}}.\label{eq:dmr
with single potential}
\end{equation}
We will apply the strength of potential functions in Equation
\ref{eq:dmr with single potential} in our following results.

\subsubsection{Strength of Potential functions in~\cite{MooijKappen_IEEETIT_07}}

\cite{MooijKappen_IEEETIT_07} mentioned a measure of the strength
of potential function $\psi_{ts}(x_t,x_s)$, which is defined as:
\begin{equation}
N(\psi_{ts})=\max_{x_t\neq\hat{x}_t,x_s\neq\hat{x}_s}\frac{\sqrt{\frac{\psi_{ts}(x_t,x_s)\psi_{ts}(\hat{x}_t,\hat{x}_s)}{\psi_{ts}(\hat{x}_t,x_s)\psi_{ts}(x_t,\hat{x}_s)}}-1}{\sqrt{\frac{\psi_{ts}(x_t,x_s)\psi_{ts}(\hat{x}_t,\hat{x}_s)}{\psi_{ts}(\hat{x}_t,x_s)\psi_{ts}(x_t,\hat{x}_s)}}+1}
=\frac{1-\sqrt{1-\sigma_{t,s}}}{1+\sqrt{1-\sigma_{t,s}}}.\label{eq:mooij
function strength}
\end{equation}
They defined log dynamic range measure as metric of errors. Let
$\lambda_{ts}$ be the log message reparameterization of message
$m_{ts}$. That is,
\[\lambda_{ts}(x_s)=\log m_{ts}(x_s).\] Denote $\Delta\lambda$ as
the difference of log messages. Thus, we have
\[\Delta\lambda_{ts}(x_s)=\log\hat{m}_{ts}(x_s)-\log
m_{ts}(x_s)=\log e_{ts}(x_s).\] By the quotient norm and Equation
(41) in~\cite{MooijKappen_IEEETIT_07}, we have the following
metric of error
\begin{equation}
\|\overline{\Delta\lambda_{ts}}\|=\frac{1}{2}\sup_{x_s,x'_s}|\Delta\lambda_{ts}(x_s)-\Delta\lambda_{ts}(x'_s)|=\log
d(e_{ts}).\label{eq:mooij error measure}
\end{equation}

Using the quotient mapping approach of parallel LBP update
in~\cite{MooijKappen_IEEETIT_07}, we will find the relationship
between the strength of potential functions in
Equation~\eqref{eq:mooij function strength} and the metric of
message errors in Equation~\eqref{eq:mooij error measure} in the
following.

Because
$\|\overline{\Delta\lambda_{ts}}\|\leq\sum_{u\in\Gamma_{t}\backslash
s}\|\overline{\frac{\partial\lambda_{ts}}{\partial\lambda_{ut}}}\|\|\overline{\Delta\lambda_{ut}}\|$
and
$\|\overline{\frac{\partial\lambda_{ts}}{\partial\lambda_{ut}}}\|\leq
N(\psi_{ts})$ by Equation (36-45)
in~\cite{MooijKappen_IEEETIT_07}, we have
\begin{eqnarray*}
&\log d(e_{ts})\leq N(\psi_{ts})\sum_{u\in\Gamma_{t}\backslash
s}\log
d(e_{ut})\leq N(\psi_{ts})\log d(E_{ts}),\\
&or,\quad d(e_{ts})\leq d(E_{ts})^{N(\psi_{ts})}.
\end{eqnarray*}
We can observe that the smaller $N(\psi_{ts})$ is, the smaller is
$d(e_{ts})$; therefore, the faster is the contraction of errors.
The previous inequality reveals another result on contractive
property of message errors beside the one in
Equation~\eqref{eq:dynamiccontraction}.

In the following, we use the maximum-error measure in
Equation~\eqref{eq:max measure} to explore upper and lower bounds
on message errors, and upper bounds on the distances between
beliefs.

\subsection{Upper- and Lower-Bounds on Message Errors}
\label{sec:upper lower bd}

We have the multiplicative error function as follows:
\begin{equation*}
e^{i+1}_{ts}(x_{s})=
   \frac{\int\psi_{ts}(x_t,x_s)M_{ts}(x_t)E^i_{ts}(x_t)dx_t}{\int\psi_{t\star}(x_t)M_{ts}(x_t)E^i_{ts}(x_t)dx_t}\times
\frac{\int\psi_{t\star}(x_t)M_{ts}(x_t)dx_t}{\int\psi_{ts}(x_t,x_s)M_{ts}(x_t)dx_t},
\end{equation*}
where $\psi_{t\star}(x_t) = \int\psi_{ts}(x_t,x_s)dx_s$. We will
show that the error function is upper- and lower- bounded.

\begin{theorem}
Multiplicative outgoing errors are bounded as:
\begin{eqnarray*}
\left(\frac{d(\psi_{ts})d(\psi_{t\star})+d(E_{ts})}{d(\psi_{ts})d(\psi_{t\star})d(E_{ts})+1}\right)^2
\leq \min_{x_s}e_{ts}(x_s) \leq e_{ts}(x_s)\leq
\max_{x_s}e_{ts}(x_s) \leq
\left(\frac{d(\psi_{ts})d(\psi_{t\star})d(E_{ts})+1}{d(\psi_{ts})d(\psi_{t\star})+d(E_{ts})}\right)^2.
\end{eqnarray*}\label{theo:error upper&lower}
\end{theorem}

The proof appears in Appendix A.

Let us use the following denotation for our upper-bound:
\begin{equation}
\Delta_1 =
\left(\frac{d(\psi_{ts})d(\psi_{t\star})d(E_{ts})+1}{d(\psi_{ts})d(\psi_{t\star})+d(E_{ts})}\right)^2.
\label{eq:up}
\end{equation}
From~\cite[Th.2 and Th.8]{ihler05b}, we can derive their
upper-bound for $\max_{x_s}e_{ts}(x_s)$:
\begin{equation}
\max_{x_s} e_{ts}(x_s) \leq
d(e_{ts})^2\leq\left(\frac{d(\psi_{ts})^2
d(E_{ts})+1}{d(\psi_{ts})^2+d(E_{ts})}\right)^2=\Delta_2.
\label{eq:up not us}
\end{equation}
\begin{theorem}
The upper bound $\Delta_1$ on the multiplicative error provided in
Theorem \ref{theo:error upper&lower} is tighter than the upper
bound $\Delta_2$ from~\cite[Th.2 and Th.8]{ihler05b}:
\label{theo:better error upper }
\end{theorem}
\begin{proof}
Because $\Delta_1$ in~\eqref{eq:up} is increasing in
$d(\psi_{t\star})$ we conclude that~\eqref{eq:up}
implies~\eqref{eq:up not us}, i.e., $\Delta_1\leq \Delta_2$,
because
\begin{align*}
&d(\psi_{t\star})
=\max_{a,b}\sqrt{\frac{\psi_{t\star}(a)}{\psi_{t\star}(b)}}
=\max_{a,b}\sqrt{\frac{\int\psi_{ts}(a,x_s)dx_s}{\int\psi_{ts}(b,x_s)dx_s}}
\\
&\leq\max_{a,b}\sqrt{\max_{c,d}\frac{\psi_{ts}(a,c)}{\psi_{ts}(b,d)}}
=\max_{a,b,c,d}\sqrt{\frac{\psi_{ts}(a,c)}{\psi_{ts}(b,d)}}=d(\psi_{ts}).
\end{align*}
\end{proof}

We can see how $d(\psi_{t\star})$ tightens the upper-bound by
analyzing the log-distance between $\Delta_1$ and $\Delta_2$. Let
$d(\psi_{t\star})=K d(\psi_{ts})$, where $1/d(\psi_{ts})\leq K\leq
1$. Therefore, the log-distance between $\Delta_1$ and $\Delta_2$ is
denoted as
\[D(K)=\log{\Delta_1}-\log{\Delta_2}=2\times\log{\{\frac{K
d(\psi_{ts})^2 d(E_{ts})+1}{K
d(\psi_{ts})^2+d(E_{ts})}\times\frac{d(\psi_{ts})^2+d(E_{ts})}{d(\psi_{ts})^2
d(E_{ts})+1}\}}.\]  We can easily find that the first gradient
$D^{(1)}(K)>0$ when $d(E_{ts})>1$. Thus, the maximum log-distance
between $\Delta_1$ and $\Delta_2$ is obtained at
$K=1/d(\psi_{ts})$. In other words, when $d(\psi_{t\star})=1$, our
upper-bound $\Delta_1$ is tighter than $\Delta_2$ at farthest.

\section{Distance Bounds on Beliefs}
\label{sec:bounds belief}

In the study of convergence, we are interested to know how beliefs
will vary at each iteration, when LBP fails to converge. We will
show that beliefs are bounded given the strength of potential
functions and the structure of the graph. In the following, we
will present our {\em uniform distance bound} and {\em non-uniform
distance bound} on beliefs.  Based on those bounds, we further
present {\em uniform convergence condition} and {\em non-uniform
convergence condition} for synchronous LBP.

\subsection{Uniform Distance Bound}

\begin{corollary}{\bf (Uniform Distance Bound)}\\
The log-distance bound of fixed points on belief at node $s$ is
\[\sum_{t\in\Gamma_s}\log(\frac{d(\psi_{ts})d(\psi_{t\star})
\varepsilon+1}{d(\psi_{ts})d(\psi_{t\star})+\varepsilon})^2,\]
where $\varepsilon$ should satisfy
\[\log{\varepsilon}=\max_{(s,p)\in
\mathbb{E}}\sum_{t\in\Gamma_s\backslash
p}\log({\frac{d(\psi_{ts})d(\psi_{t\star})\varepsilon+1}{d(\psi_{ts})d(\psi_{t\star})+\varepsilon}})^2.\]\label{coro:uniform
distance bd}
\end{corollary}

The proof appears in Appendix A.

Let us reintroduce the {\em error bound-variation function} used
in the proof for Corollary \ref{coro:uniform distance bd}:
\begin{equation}
G_{sp}^{O}(\log{\varepsilon})=\log{\prod_{t\in\Gamma_s\backslash
p}(\frac{d(\psi_{ts})d(\psi_{t\star})\varepsilon+1}{d(\psi_{ts})d(\psi_{t\star})+\varepsilon}})^2-\log{\varepsilon},\varepsilon\geq
1.\label{eq:bd variation}
\end{equation}
Adopting the upper-bound $\Delta_2$ in \eqref{eq:up not us}, the
error bound-variation function is:
\begin{equation}
G_{sp}^{I}(\log{\varepsilon})=\log{\prod_{t\in\Gamma_s\backslash
p}(\frac{d(\psi_{ts})^2\varepsilon+1}{d(\psi_{ts})^2+\varepsilon}})^2-\log{\varepsilon},\varepsilon\geq
1.
\end{equation}

Those error bound-variation functions describe the upper-bound on
variation of maximal message errors throughout the belief
networks. We can see that
$G_{sp}^{O}(\log{\varepsilon})<G_{sp}^{I}(\log{\varepsilon})$. In
other words, the error bound-variation function using our
upper-bound $\Delta_1$ is tighter than that using
\cite{ihler05b}'s upper-bound $\Delta_2$, which is illustrated in
Fig.~\ref{fig:bd variation}. However, in \cite{ihler05b}, they
used the following error bound-variation function:
\begin{equation}
G_{sp}^{II}(\log{\varepsilon'})=\log{\prod_{t\in\Gamma_s\backslash
p}(\frac{d(\psi_{ts})^2\varepsilon'+1}{d(\psi_{ts})^2+\varepsilon'}})-\log{\varepsilon'},
\end{equation}
where $\varepsilon'$ is an upper-bound on dynamic range measure
$d(E_{ts})$. Since our $\varepsilon$ is an upper-bound on maximum
error measure $\max{E_{ts}}$, it's hard to compare
$G_{sp}^{O}(\log{\varepsilon})$ and
$G_{sp}^{II}(\log{\varepsilon'})$. In other words, we cannot say
our Uniform Distance Bound in Corollary \ref{coro:uniform distance
bd} is better than that in~\cite[Theorem 13]{ihler05b}.

\begin{figure}[htb]
\begin{center}
\includegraphics[height=7.44 cm,width=14.02 cm]{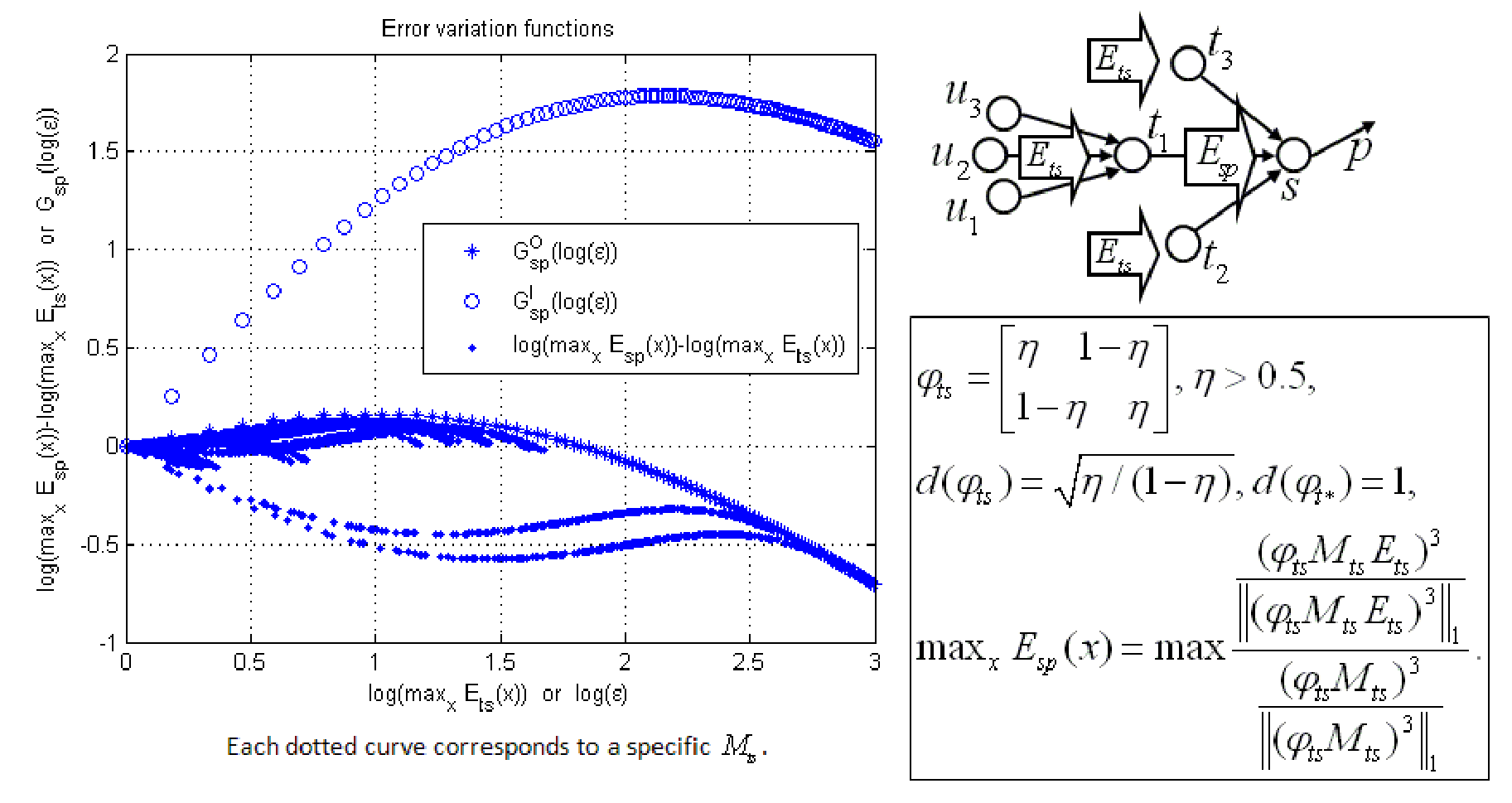}
\end{center}
\caption{Error bound-variation functions versus true
error-variation function for the local graph of node s. Potential
functions on edges $(t_1,s),(t_2,s),(t_3,s)$ are the same, where
$\eta=0.7$. We also impose the same incoming error product
$E_{ts}$ on nodes $t_1,t_2,t_3$. The dotted curves depict the true
error variation functions, $\{\log{\max_x E_{sp}(x)}-\log{\max_x
E_{ts}(x)},t\in\Gamma_s\backslash p\}$, which are enveloped by our
error bound-variation function $G_{sp}^{O}(\log{\varepsilon})$.}
\label{fig:bd variation}
\end{figure}
\begin{figure}[htb]
\begin{center}
\includegraphics[height=3.91 cm,width=10.25 cm]{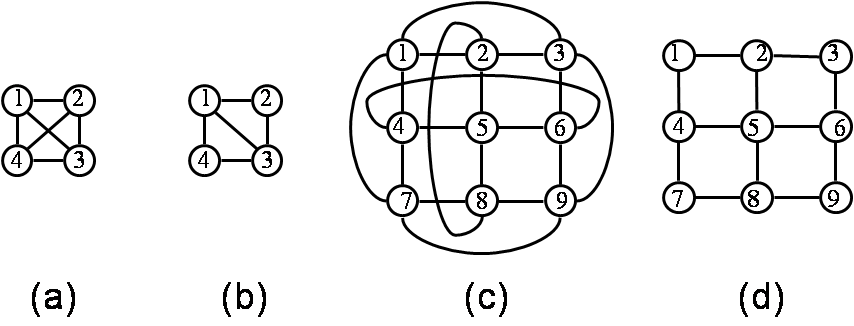}
\end{center}
\caption{Four simple graphical models: (a) a four-node fully
connected graph; (b) a partial graph that has one less edge than
(a); (c) a nine-node graph with uniform degree; and (d) a $3\times
3$ grid that is a partial graph of (c).} \label{fig:graphs}
\end{figure}

When the error bound-variation function is always less than zero,
the maximum of error bounds decreases after each iteration of LBP.
In other words, LBP will converge.  Therefore, our uniform
distance bound in Corollary \ref{coro:uniform distance bd} will
lead to a sufficient condition for convergence of LBP.
\begin{theorem}{\textbf{ (Uniform Convergence Condition)}}\\
Based on maximum-error measure, the sufficient condition for the
sum-product algorithm to converge to a unique fixed point is
\[\max_{(s,p)\in\mathbb{E}}\sum_{t\in\Gamma_s\backslash
p}\frac{d(\psi_{ts})d(\psi_{t\star})-1}{d(\psi_{ts})d(\psi_{t\star})+1}<\frac{1}{2}.\]\label{theo:uniform
converge}
\end{theorem}

The proof appears in Appendix A.

Since we cannot compare $G_{sp}^{O}(\log{\varepsilon})$ and
$G_{sp}^{II}(\log{\varepsilon'})$ directly because $\varepsilon$
and $\varepsilon'$ correspond to different measures, let us take
the maximum of the two measures and deal with it as a new measure.
Specifically, let
$\tilde{\varepsilon}=\max\{\varepsilon,\varepsilon'\}$. After some
calculation, we can find that
$G_{sp}^{O}(\log{\tilde{\varepsilon}})$ is greater than
$G_{sp}^{II}(\log{\tilde{\varepsilon}})$. In other words,
$G_{sp}^{II}(\log{\tilde{\varepsilon}})$ is tighter than
$G_{sp}^{O}(\log{\tilde{\varepsilon}})$. Therefore, the
convergence condition derived from
$G_{sp}^{II}(\log{\tilde{\varepsilon}})$ will be better. The
following lemma provides a proof for this observation.

\begin{lemma}
Our sufficient condition $\sum_{t\in\Gamma_s\backslash
p}\frac{d(\psi_{ts})d(\psi_{t\star})-1}{d(\psi_{ts})d(\psi_{t\star})+1}<\frac{1}{2}$
is worse than the sufficient condition in~\cite{ihler05b}, which
is $\sum_{t\in\Gamma_s\backslash
p}\frac{d(\psi_{ts})^2-1}{d(\psi_{ts})^2+1}<1$.\label{lemma:weak
convergence}
\end{lemma}
\begin{proof}
$2(\frac{d(\psi_{ts})d(\psi_{t\star})-1}{d(\psi_{ts})d(\psi_{t\star})+1})>\frac{d(\psi_{ts})^2-1}{d(\psi_{ts})^2+1}$.
\end{proof}

Our failure to improve the uniform convergence condition by using
{\em maximum-error measure} shows that {\em dynamic-range measure}
is better than {\em maximum-error measure} with respect to the
sensitivity of the measure to convergence. Nevertheless, as for
the upper bound on a multiplicative message error $e_{ts}(x)$,
{\em maximum-error measure} gives a tighter result, which is shown
in Theorem \ref{theo:better error upper }. Furthermore, the {\em
maximum-error measure} may provide better distance bounds for
beliefs.

Inspired by the sensitivity of dynamic-range measure to
convergence, we present the following {\em improved uniform
distance bound}, which first calculates the fixed-point values of
error bounds in dynamic-range measure, and then computes the error
bounds among beliefs in maximum-error measure.

\begin{corollary}{\bf (Improved Uniform Distance Bound)}\\
The log-distance bound of fixed points on belief at node $s$ is
\[\sum_{t\in\Gamma_s}\log(\frac{d(\psi_{ts})d(\psi_{t\star})
\varepsilon+1}{d(\psi_{ts})d(\psi_{t\star})+\varepsilon})^2,\]
where $\varepsilon$ should satisfy
\[\log{\varepsilon}=\max_{(s,p)\in
\mathbb{E}}\sum_{t\in\Gamma_s\backslash
p}\log{\frac{d(\psi_{ts})^2\varepsilon+1}{d(\psi_{ts})^2+\varepsilon}}.\]\label{coro:improved
uniform distance}
\end{corollary}
\begin{proof}
Using the approach in~\cite[Theorem 12]{ihler05b} to obtain
distance bounds on incoming error products in dynamic-range
measure and applying our Theorem \ref{theo:error upper&lower}, we
obtain our corollary.
\end{proof}

\begin{figure}[htb]
\begin{center}
\includegraphics[height=10.1 cm,width=14 cm]{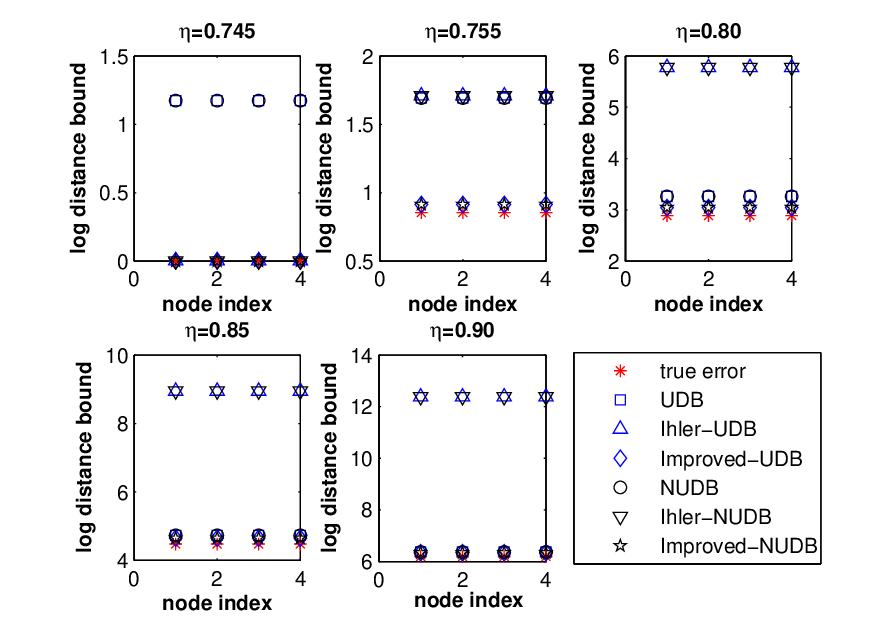}
\end{center}
\caption{True distance, uniform distance bounds and non-uniform
distance bounds for the graph in ~\ref{fig:graphs}(a) with various
$\eta$'s. The empirical critical value of $\eta$ for LBP to
converge is $\eta<0.75$.} \label{fig:bounds fig3a}
\end{figure}
\begin{figure}[htb]
\begin{center}
\includegraphics[height=10.33 cm,width=14 cm]{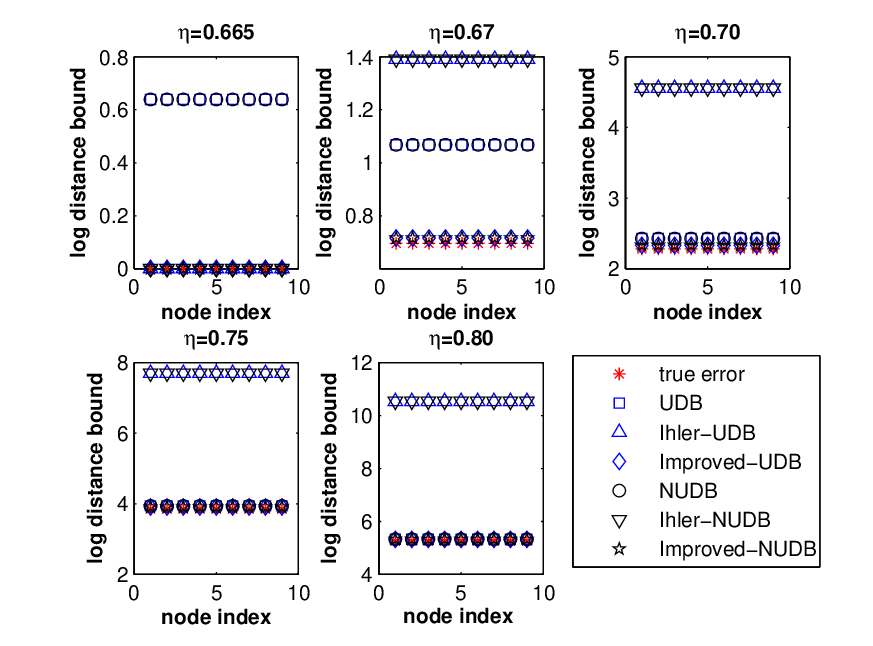}
\end{center}
\caption{True distance, uniform distance bounds and non-uniform
distance bounds for the graph in ~\ref{fig:graphs}(c) with various
$\eta$'s. The empirical critical value of $\eta$ for LBP to
converge is $\eta<0.67$.} \label{fig:bounds fig3c}
\end{figure}
\begin{figure}[htb]
\begin{center}
\includegraphics[height=10.1 cm,width=14 cm]{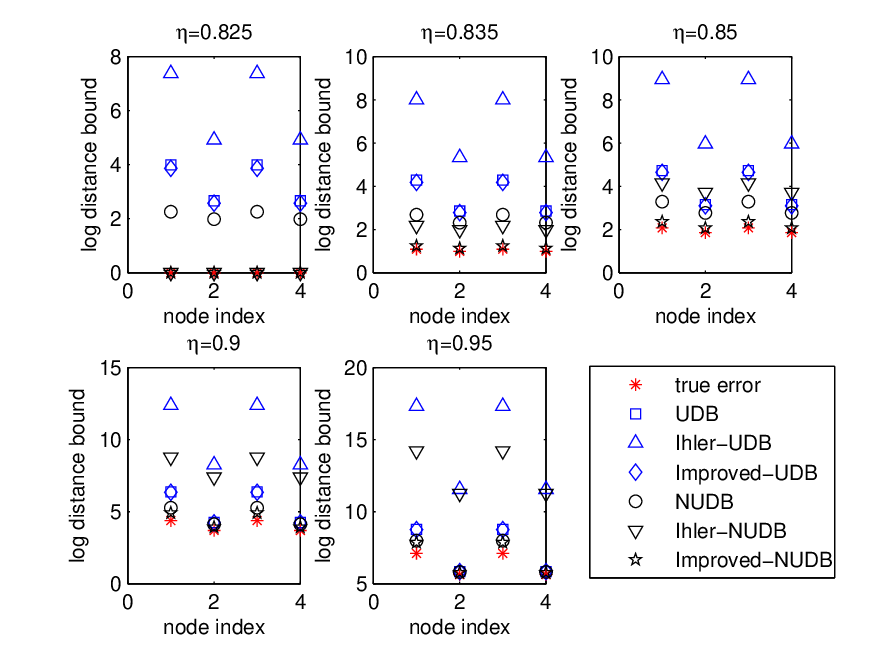}
\end{center}
\caption{True distance, uniform distance bounds and non-uniform
distance bounds for the graph in ~\ref{fig:graphs}(b) with various
$\eta$'s. The empirical critical value of $\eta$ for LBP to
converge is $\eta<0.83$.} \label{fig:bounds fig3b}
\end{figure}
\begin{figure}[htb]
\begin{center}
\includegraphics[height=10.33 cm,width=14 cm]{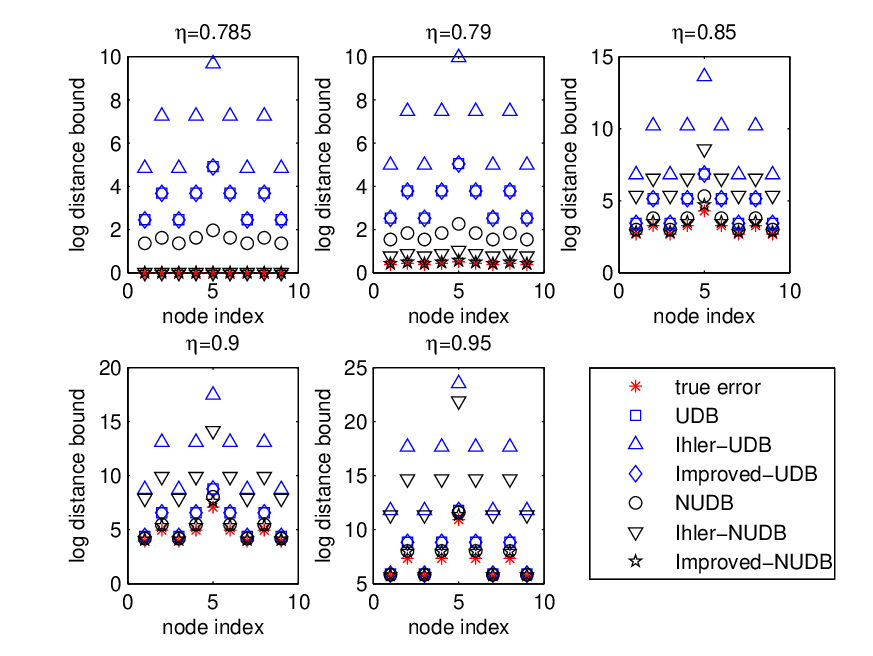}
\end{center}
\caption{True distance, uniform distance bounds and non-uniform
distance bounds for the graph in ~\ref{fig:graphs}(d) with various
$\eta$'s. The empirical critical value of $\eta$ for LBP to
converge is $\eta<0.79$.} \label{fig:bounds fig3d}
\end{figure}
Let see how our {\em uniform distance bound} and {\em improved
uniform distance bound} perform for graphical models in
Fig.~\ref{fig:graphs} by comparison to the {\em Fixed-point
distance bound} in~\cite{ihler05b}. Let all the pairwise potential functions be $\begin{pmatrix} \eta & 1-\eta \\
1-\eta & \eta \end{pmatrix}$ where $\eta>0.5$ and all the single
node potentials be $\begin{pmatrix} 1 \\ 1 \end{pmatrix}$.
Therefore, $d(\psi_{ts})=\sqrt{\eta/(1-\eta)}$ and
$d(\psi_{t\star})=1$ for $\forall~(t,s)\in \mathbb{E}$.

We compare the following bounds in our simulations: UDB, our
uniform distance bound in Corollary \ref{coro:uniform distance
bd}; Improved-UDB, our improved uniform distance bound in
Corollary \ref{coro:improved uniform distance}; Ihler-UDB,
Fixed-point distance bound in~\cite[Theorem 13]{ihler05b}.
Fig.\ref{fig:bounds fig3a} - Fig.~\ref{fig:bounds fig3d}
illustrate the performances of those bounds for graphs in
Figs.~\ref{fig:graphs}(a), (c), (b) and (d), respectively.

Graphs in Figs.~\ref{fig:graphs}(a) and (c) are uniform (uniform
degrees, uniform potential functions). Given a specific $\eta$,
all nodes have the same distance bound. The \textbf{critical
value} of $\eta$ is the value beyond which LBP will not converge.
For those two graphs, the empirical critical values of $\eta$ with
respect to the convergence of LBP are $0.75$ and $0.67$
respectively. We can see that, for various $\eta$'s, our
Improved-UDBs are very close to the true errors between beliefs.
Our UDBs become tighter when $\eta$ increases, while Ihler-UDBs
become looser. From Fig.~\ref{fig:bounds fig3a} and
Fig.~\ref{fig:bounds fig3c}, we can see that, compared to
Ihler-UDB, our UDB requires stricter critical values of $\eta$ to
ensure error bounds to be zeros. Specifically, for
Fig.~\ref{fig:bounds fig3a}, when $\eta=0.745$, our UDBs are
non-zeros and Ihler-UDBs are zeros; hence, our UDB requires
$\eta<0.745$ for the convergence of LBP, while Ihler-UDB only
requires $\eta<0.75$. Nevertheless, the critical values by our UDB
are $0.735$ for Fig.~\ref{fig:graphs}(a) and $0.66$ for
Fig.~\ref{fig:graphs}(c), which are close to the empirical
critical values. Based on our UDB and Ihler-UDB, our Improved-UDBs
will approximate zeros when $\eta$ approaches $0.75$ and give
tightest distance bounds for any $\eta$.

\subsection{Non-Uniform Distance Bound}
\label{sec:nonuniform bd}

Fig.~\ref{fig:graphs}(b) and Fig.~\ref{fig:graphs} (d) are
non-uniform graphs. Because uniform distance bounds are computed
locally, beliefs on the nodes with different topologies will have
different error bounds, which can be observed from
Fig.~\ref{fig:bounds fig3b} and Fig.~\ref{fig:bounds fig3d}. We
can also find that when the true errors are zeros, uniform bounds
are not all zeros. In other words, $\eta$ must be smaller than the
empirical critical value to ensure the largest uniform distance
bounds to be zero. Furthermore, in such cases, uniform convergence
conditions derived from uniform distance bounds will not perform
well as for uniform graphs. Therefore, when every loop contains
potentials with various strengths and each node has different
topology, we present the following {\em non-uniform distance
bound} and {\em improved non-uniform distance bound}.
\begin{corollary}{\textbf{ (Non-uniform Distance Bound)}}\\
The non-uniform log-distance bound of fixed points on belief at
node $s$ after $n\geq 1$ iterations is
\[\sum_{t\in\Gamma_s}\log(\frac{d(\psi_{ts})d(\psi_{t\star})
\varepsilon^n_{ts}+1}{d(\psi_{ts})d(\psi_{t\star})+\varepsilon^n_{ts}})^2,\]
where $\varepsilon^i_{ts}$ is updated by
\[\log{\varepsilon^{i}_{ts}}=\sum_{u\in\Gamma_t\backslash
s}\log({\frac{d(\psi_{ut})d(\psi_{u\star})\varepsilon^{i-1}_{ut}+1}{d(\psi_{ut})d(\psi_{u\star})+\varepsilon^{i-1}_{ut}}})^2\]
with initial condition
\[\log\varepsilon^1_{ut}=\sum_{v\in\Gamma_u\backslash t}\log
(d(\psi_{vu})d(\psi_{v\star}))^2.\]\label{coro:nonuniform
distance}
\end{corollary}
\begin{proof}
The result can be easily proved from Corollary \ref{coro:uniform
distance bd}, by defining the {\em error bound-variation function}
in \eqref{eq:bd variation} as follows:
\begin{equation*}
G_{ts}(\log{\varepsilon_{ts}^{i}})=\log\prod_{u\in\Gamma_t\backslash
s}\Delta_{ut}(\varepsilon_{ut}^{i-1})-\log{\varepsilon_{ts}^{i}}=\sum_{u\in\Gamma_t\backslash
s}\log({\frac{d(\psi_{ut})d(\psi_{u\star})\varepsilon^{i-1}_{ut}+1}{d(\psi_{ut})d(\psi_{u\star})+\varepsilon^{i-1}_{ut}}})^2-\log{\varepsilon^{i}_{ts}}.
\end{equation*}

\end{proof}
Similarly, based on the fact that the dynamic-range measure gives
better convergence condition than the maximum-error measure, we
improve the previous non-uniform distance bound in the following.
\begin{corollary}{\bf (Improved Non-uniform Distance Bound)}\\
The improved non-uniform log-distance bound of fixed points on
belief at node $s$ after $n\geq 1$ iterations is
\[\sum_{t\in\Gamma_s}\log(\frac{d(\psi_{ts})d(\psi_{t\star})
\varepsilon^n_{ts}+1}{d(\psi_{ts})d(\psi_{t\star})+\varepsilon^n_{ts}})^2,\]
where $\varepsilon^i_{ts}$ is updated by
\[\log{\varepsilon^{i}_{ts}}=\sum_{u\in\Gamma_t\backslash
s}\log{\frac{d(\psi_{ut})^2\varepsilon^{i-1}_{ut}+1}{d(\psi_{ut})^2+\varepsilon^{i-1}_{ut}}}\]
with initial condition
$\log\varepsilon^1_{ut}=\sum_{v\in\Gamma_u\backslash t}\log
d(\psi_{vu})^2$.\label{coro:improved nonuniform distance}
\end{corollary}
\begin{proof}
Using the approach in~\cite[Theorem 14]{ihler05b} to obtain
distance bounds on incoming error products in dynamic-range
measure and applying our Theorem \ref{theo:error upper&lower}, we
obtain our corollary.
\end{proof}

Let see the performaces of our {\em non-uniform distance bound}
and {\em improved non-uniform distance bound} for the graphs in
Fig.~\ref{fig:graphs} compared with the non-uniform distance bound
in~\cite[Thm. 14]{ihler05b}. We denote the bounds in our
simulation as follows: NUDB, our non-uniform distance bound in
Corollary \ref{coro:nonuniform distance}; Improved-NUDB, our
improved non-uniform distance bound in Corollary
\ref{coro:improved nonuniform distance}; Ihler-NUDB, non-uniform
distance bound in~\cite[Theorem 14]{ihler05b}.

For uniform graphs in Fig.~\ref{fig:graphs}(a) and (c), NUDB
performs exactly the same as UDB. However, for non-uniform graphs
in Fig.~\ref{fig:graphs}(b) and (d), because NUDB propagates error
bounds throughout the whole graph rather than on a local
neighborhood, NUDBs are tighter than UDBs, which can be observed
from Fig.~\ref{fig:bounds fig3b} and Fig.~\ref{fig:bounds fig3d}.
For various $\eta$'s, our Improved-NUDBs always approach the true
errors. Therefore, when our Improved-NUDB is zero, $\eta$ almost
equals the empirical critical value to ensure convergence of LBP.
Though worse than Improved-NUDB, our NUDB performs better than
Ihler-NUDB when $\eta$ is far way from the area of convergence.

\subsubsection{Non-Uniform Convergence}
\label{sec:nonuniform convergence}

Based on our Improved-NUDB or Ihler-NUDB, a sufficient convergence
condition of LBP can be derived, which is based on the
dynamic-range measure of propagating errors.

For each cycle-involved vertex $v$, $T(\mathbb{G},v)$ is the
corresponding computation tree. Let $\mathbb{V}$ be the set of
vertices in the computation tree. For $w_{i}\in
\mathbb{V},i=0,...,|\mathbb{V}|-1$, $l(w_{i})$ is the labelling
function which maps $w_{i}$ to the original vertex in
$\mathbb{G}$. Let $l(w_0)=v$.

\begin{theorem}{\bf (Non-Uniform Convergence Condition)}\\
For a graphical model $\mathbb{G}(\mathbb{V},\mathbb{E})$,
$\{T(\mathbb{G},v), v\in\mathbb{V}\}$ is the set of computation
trees. Let $\mathbb{\bar{E}}$ denote the set of directed edges.
For each $T(\mathbb{G},v), v\in\mathbb{V}$, given
$vu\in\mathbb{\bar{E}}$, $\mathcal{H}_{vu}$ denotes an expression
on edge $vu$:
\begin{equation}
\mathcal{H}_{vu}=\sum_{w_i\in\Gamma_{v}\backslash
u}\frac{d(\psi_{l(w_i)v})^2-1}{d(\psi_{l(w_i)v})^2+1}\sum_{w_j\in\Gamma_{w_i}\backslash
v}\frac{d(\psi_{l(w_j)l(w_i)})^2-1}{d(\psi_{l(w_j)l(w_i)})^2+1}...
\sum_{w_r\in\Gamma_{w_q}\backslash
w_p}\frac{d(\psi_{l(w_r)l(w_q)})^2-1}{d(\psi_{l(w_r)l(w_q)})^2+1},\label{eq:non-uniform
walk sums}
\end{equation}
where $\Gamma_{w_i}$ is the set of neighbors of $w_i$. The
non-uniform sufficient condition for the sum-product algorithm to
converge to a local stable fixed point is:
\begin{equation*}
\max_{vu\in\mathbb{\bar{E}}}\mathcal{H}_{vu}<1.
\end{equation*}
\label{theo:nonuniform convergence}
\end{theorem}

The proof appears in Appendix A. Based on the type of computation
tree, the non-uniform convergence condition will be called
\textit{non-uniform convergence condition based on $N$-th level
Bethe tree}, or \textit{non-uniform convergence condition based on
infinite Bethe tree}, or \textit{non-uniform convergence condition
based on SAW tree}. Our non-uniform convergence condition based on
infinite Bethe tree is equivalent to \cite[Theorem 14]{ihler05b}.

When a graph has uniform potential functions with strength
$d(\psi)$, to ensure convergence, it is sufficient to have
\begin{equation}
\max_{vu\in\mathbb{\bar{E}}}\sum_{w_i\in\Gamma_{v}\backslash
u}\frac{d(\psi)^2-1}{d(\psi)^2+1}\sum_{w_j\in\Gamma_{w_i}\backslash
v}\frac{d(\psi)^2-1}{d(\psi)^2+1}...\sum_{w_r\in\Gamma_{w_q}\backslash
w_p}\frac{d(\psi)^2-1}{d(\psi)^2+1}<1.\label{eq:walk sums uniform}
\end{equation}

Let us apply our {\em non-uniform convergence condition based on
SAW tree} to the graphs in Fig.~\ref{fig:graphs}(b) and (d) with
uniform potential functions as in the previous simulations. For
the graph in Fig.~\ref{fig:graphs}(b), we obtain the critical
value $\eta<0.78$ for convergence of LBP, which is closer to the
empirical value $\eta<0.83$, compared to $\eta<0.75$ obtained by
{\em uniform convergence condition}. For the graph in
Fig.~\ref{fig:graphs}(d), we obtain the critical value
$\eta<0.77$, while the empirical value is $\eta<0.79$ and the
critical value obtained by {\em uniform convergence condition} is
$\eta<0.67$. Therefore, our {\em non-uniform convergence
condition} is tighter than our {\em uniform convergence
condition}. However, since our {\em non-uniform convergence
condition} is derived from \cite[Theorem 14]{ihler05b}, we do not
improve the convergence condition. Rather than in the form of
distance bound in \cite[Theorem 14]{ihler05b}, we express the
convergence condition explicitly, which will be used in our later
analysis of walk-summability of graphical models. Furthermore, we
improve distance bounds between beliefs in Corollary
\ref{coro:improved uniform distance} and Corollary
\ref{coro:improved nonuniform distance}, which are useful in
tightening accuracy bounds in Section \ref{sec:accuracy}.

\section{Convergence of Loopy Belief Propagation}
\label{sec:LBP convergence}
\subsection{Sparsity and Convergence}
It lacks theoretical verification that the more sparse a graph is,
the less stricter is its convergence condition. Since the
definition of sparse graphs is vague, to be confined, we would
relate sparsity with partial graphs. In this section, we will show
that when LBP converges on one graph in a partial graph set,
convergence properties of other graphs can be deduced through our
Theorem \ref{theo:nonuniform convergence}. Let us define partial
graphs and introduce the convergence property of such graphs in
the following.
\begin{definition}{\bf (Walk)}\\
In a graph G(V,E), a walk of length $l$ is a sequence of nodes
$w=(v_0,v_1,...,v_l)$, $v_i\in V$, such that each step of walk
$(v_i,v_i+1)$ corresponds to an edge in $E$.\label{def:walk}
\end{definition}
\begin{definition}{\bf (Prime Cycle)}\\
A closed walk is called a prime cycle if it is not backtracking
and not a repeated concatenation of a shorter closed
walk.\label{def:prime cycle}
\end{definition}
\begin{definition}{\bf (Reduction)}\\
A walk composed of two edges $(v_1,v_2)$ and $(v_2,v_3)$ can be
reduced to a walk composed of one edge $(v_1,v_3)$, where
$\psi_{v_1v_3}(x_{v_1},x_{v_3})=\int_{x_{v_2}}\psi_{v_1v_2}(x_{v_1},x_{v_2})\psi_{v_2v_3}(x_{v_2},x_{v_3})dx_{v_2}$,
when there is no branch on the walk.
\end{definition}
\begin{definition}{\bf (Extension)}\\
A walk composed of one edge $(v_1,v_3)$ can be extended to a walk
composed of two edges $(v_1,v_2)$ and $(v_2,v_3)$, where
$\int_{x_{v_2}}\psi_{v_1v_2}(x_{v_1},x_{v_2})\psi_{v_2v_3}(x_{v_2},x_{v_3})dx_{v_2}=\psi_{v_1v_3}(x_{v_1},x_{v_3})$.
\end{definition}
It is not hard to prove that {\bf {Reduction}} and {\bf
{Extension}} do not change the convergence property of the
original graph. Comparatively, \cite{Ruozzi2010} splitted some
edges and reparameterized the original graphical model in order to
obtain a convergent and correct message passing algorithm.
\begin{definition}{\bf (Partial Graphs)}\\
For two graphical models $\mathbb{G}_1(\mathbb{V}_1,\mathbb{E}_1)$
and $\mathbb{G}_2(\mathbb{V}_2,\mathbb{E}_2)$ after reduction and
extension, there exists an isomorphism between graphs
$\mathbb{G}_1(\mathbb{V}_1,\mathbb{E}_1)$ and
$\mathbb{G}_2(\mathbb{V}_2^{\ast},\mathbb{E}_2^{\ast})$, when
$\mathbb{V}_2^{\ast}\subseteq\mathbb{V}_2$ and
$\mathbb{E}_2^{\ast}\subset\mathbb{E}_2$. When
$\mathbb{E}_2-\mathbb{E}_2^{\ast}$ is cycle-involved, we call
$\mathbb{G}_1$ a partial graph of $\mathbb{G}_2$ and denote it as
$\mathbb{G}_1\subset\mathbb{G}_2$.\label{def:partial graph}
\end{definition}
\begin{theorem}{\bf{ (Strictness of Convergence Condition for Two Partial Graphs)}}\\
Given $\mathbb{G}_1$ and $\mathbb{G}_2$ as defined in Definition
\ref{def:partial graph}, assume that
$\mathbb{G}_1\subset\mathbb{G}_2$. Assume the dynamic-range
measures of potential functions for edges in $\mathbb{E}_1$ are
not greater than those of potential functions for corresponding
edges in $\mathbb{E}_2^{\ast}$. Then, when LBP for
$\mathbb{G}_2(\mathbb{V}_2,\mathbb{E}_2)$ converges, LBP for
$\mathbb{G}_1(\mathbb{V}_1,\mathbb{E}_1)$ must converge; however,
the reverse implication is not true in
general.\label{theo:strictness convergence}
\end{theorem}
\begin{proof}
Because $\mathbb{G}_1\subset \mathbb{G}_2$ and
$\mathbb{E}_2-\mathbb{E}_2^{\ast}$ are cycle-involved,
$T_{B}(\mathbb{G}_1,v,n)\subset T_{B}(\mathbb{G}_2,v,n)$.
Therefore, the expression in~\eqref{eq:non-uniform walk sums} for
$\mathbb{G}_2$ has more summands than that for $\mathbb{G}_1$.
When $\mathbb{G}_2$ satisfies the convergence condition in Theorem
\ref{theo:nonuniform convergence}, $\mathbb{G}_1$ must satisfy it.
However, when $\mathbb{G}_1$ satisfies the convergence condition,
$\mathbb{G}_2$ may not satisfy it.
\end{proof}
When the potential functions of a graph are uniform, we have the
following corollary.
\begin{corollary}{\bf (Critical Values of Convergence for Two Partial Graphs)}\\
Given $\mathbb{G}_1\subset\mathbb{G}_2$, $\mathbb{G}_1$ and
$\mathbb{G}_2$ have uniform potential
functions $\psi_{i}=\begin{pmatrix} \eta_i & 1-\eta_i \\
1-\eta_i & \eta_i \end{pmatrix}, i=1,2$ on all edges. Then, the
critical values for convergence of LBP satisfy
$\eta_{2}<\eta_{1}$.\label{coro:critical value convergence}
\end{corollary}
\begin{proof}
Because \eqref{eq:walk sums uniform} for $\mathbb{G}_2$ has more
summands than that for $\mathbb{G}_1$, we easily have
$d(\psi_{2})<d(\psi_{1})$ to satisfy the inequality. Because
$d(\psi_i)=\sqrt{\eta_i/(1-\eta_i)}$, we get $\eta_{2}<\eta_{1}$.
\end{proof}
Our Theorem \ref{theo:strictness convergence} and Corollary
\ref{coro:critical value convergence} can be easily extended to
strictness of convergence condition of LBP for a set of partial
graphs, and for those with uniform potential functions.
\begin{corollary}{\bf{ (Strictness of Convergence Condition for Set of Partial Graphs)}}\\
Given $\mathbb{G}_1\subset\mathbb{G}_2...\subset\mathbb{G}_N$,
assuming the dynamic-range measures of potential functions on
isomorphous edges of those graphs are correspondingly
non-decreasing in the previous partial order, LBP convergence for
$\mathbb{G}_j$ implies LBP convergence for $\mathbb{G}_i$, where
$i<j$ and $i,j=1,...,N$. However, the reverse implication is not
true in general.\label{coro:strictness convergence set partial
graphs}
\end{corollary}
\begin{proof}
For any $\mathbb{G}_i\subset\mathbb{G}_j$ in the set of
$\{\mathbb{G}_i,1\leq i\leq N\}$, according to Theorem
\ref{theo:strictness convergence}, we have the convergence of
$\mathbb{G}_j$ implies the convergence of $\mathbb{G}_i$.
\end{proof}
\begin{corollary}{\bf (Critical Value of Convergence for Set of Partial Graphs)}\\
Given $\mathbb{G}_1\subset\mathbb{G}_2...\subset\mathbb{G}_k$,
$\mathbb{G}_1$,..., $\mathbb{G}_k$ have uniform potential
functions $\begin{pmatrix} \eta_i & 1-\eta_i \\
1-\eta_i & \eta_i \end{pmatrix},1\leq i\leq k$ on all edges. Then,
the critical values for convergence of LBP satisfy
$\eta_{k}<\eta_{k-1}...<\eta_{1}$.\label{coro:critical value
convergence set partial graphs}
\end{corollary}
\begin{proof}
For any $\mathbb{G}_i\subset\mathbb{G}_j$ in the set of
$\{\mathbb{G}_i,1\leq i\leq N\}$, according to Corollary
\ref{coro:critical value convergence}, we have the convergence of
$\eta_{j}<\eta_{i}$.
\end{proof}

By our Corollary \ref{coro:strictness convergence set partial
graphs} on partially ordered graphs, we can conclude that graphs
with less cycle-induced edges are more sparse and thus have weaker
convergence condition. It is intuitively true that the strength of
potential functions for Fig.~\ref{fig:graphs}(a) or
Fig.~\ref{fig:graphs}(c) should be weaker than that for
Fig.~\ref{fig:graphs}(b) or Fig.~\ref{fig:graphs}(d) to ensure
convergence of LBP. This observation can be soundly verified by
our previous corollaries.

\subsection{Walk-Summability and Convergence}
\begin{figure}[htb]
\begin{center}
\includegraphics[height=4.5 cm,width=9.93 cm]{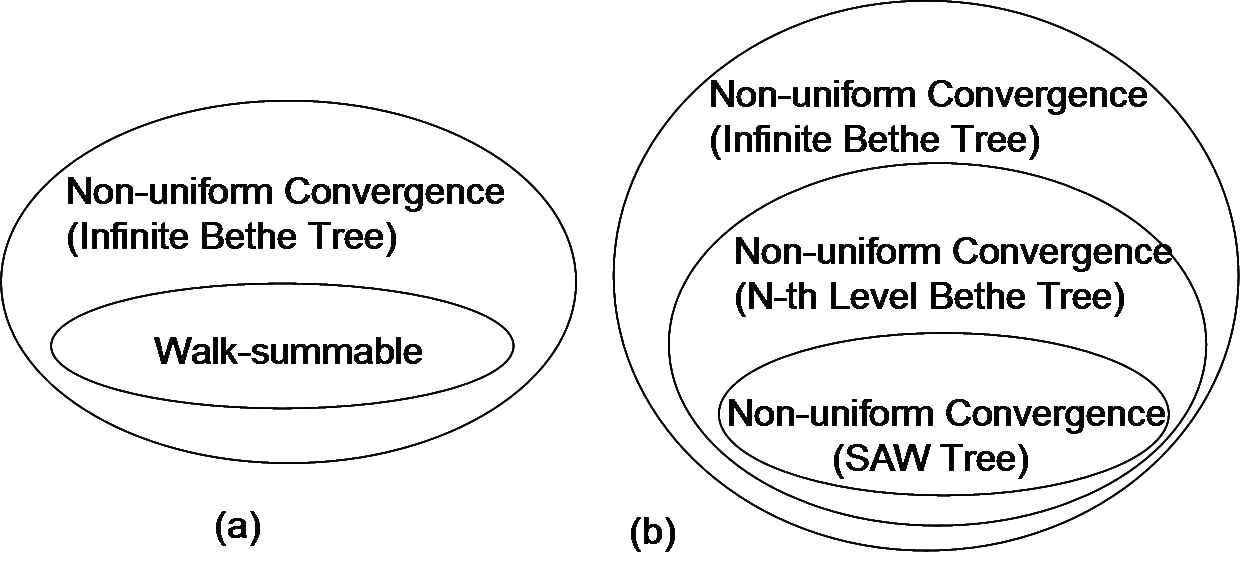}
\end{center}
\caption{Diagram summarizing mildness of convergence conditions.
The SAW tree is a partial tree of the $N$-level Bethe tree,
therefore, convergence condition based on the SAW tree is
stronger.} \label{fig:tightness_nonuniform}
\end{figure}
\cite{Malioutov2006} related the convergence of LBP with the
spectral radius of partial correlation matrix of Gaussian
graphical model, for which they introduced a concept called
walk-summability. We observe similarity between walk-summability
of Gaussian graphical model and our convergence condition for
general graphcial model discussed in Section \ref{sec:nonuniform
convergence}. Therefore, based on some existing works in
literature, we extend the walk-summability defined
in~\cite{Malioutov2006} to that for general graphical models.

A Gaussian graphical model is defined by an undirected graph
$G(V,E)$, where $V$ is the set of nodes and $E$ is the set of
edges, and a set of jointly Gaussian random variables $\{x_i, i\in
V\}$. The joint density function is defined as follows:
\begin{equation*}
p(X)\propto
\exp\{-\frac{1}{2}\textsc{x}^TJ\textsc{x}+h^T\textsc{x}\},
\end{equation*}
where $J$ is a symmetric and positive definite matrix called
information matrix and $h$ is a potential vector. The partial
correlation coefficient between random variable $x_i$ and $x_j$ is
defined as follows:
\begin{equation*}
r_{ij}\triangleq\frac{\textrm{cov}(x_i,x_j|x_{V\backslash
ij})}{\sqrt{\textrm{var}(x_i|x_{V\backslash
ij})\textrm{var}(x_j|x_{V\backslash
ij})}}=-\frac{J_{ij}}{\sqrt{J_{ii}J_{jj}}}.
\end{equation*}
A walk is defined in Definition \ref{def:walk}. The weight
$\phi(w)$ of a walk $w=(v_0,v_1,...,v_{l(w)})$ with length $l(w)$
is defined as:
\begin{equation}
\phi(w)=\prod_{k=1}^{l(w)}r_{v_{k-1}v_{k}}.\label{eq:walk weight}
\end{equation}
\begin{definition}(\cite{Malioutov2006}){\bf (Walk-Summable)}\\
A Gaussian distribution is walk-summable if for all $i,j\in V$ the
unordered walk $w$ from $i$ to $j$, $\sum_{w:i\rightarrow
j}\phi(w)$, is well defined.\label{def:walk-summable}
\end{definition}
\begin{proposition}(\cite{Malioutov2006}){\bf (Walk-Summability)}\\
Let $R$ be a partial correlation coefficient matrix of a Gaussian
graphical model, of which diagonal entries are zeros. Each of the
following conditions are equivalent to walk-summability:\\
(i) $\sum_{w:i\rightarrow j}|\phi(w)|$ converges for all $i,j\in
V$,\\
(ii) $\sum_{l}\bar{R}^{l}$ converges, where
$\bar{R}_{ij}=|R_{ij}|$
and $l$ is the length of walk,\\
(iii) $\rho(\bar{R})<1$, where $\rho(\bar{R})$ is the spectral radius of $\bar{R}$,\\
(iv) $I-\bar{R}\succ 0$.
\end{proposition}

The walk-summability of a Gaussian graphical model has been shown
to be related with the convergence of LBP. Proposition 21
in~\cite{Malioutov2006} states that ``If a model on a (Gaussian)
graph G is walk-summable, then LBP is well-posed, the means
converge to the true means and the LBP variances converge to
walk-sums over the backtracking self-return walks at each node".
Enlightened by the analysis for Gaussian graphical model, we
extend the walk-summability perspective to general graphical
models in the following.

For a Gaussian graphical model, the interaction between two random
variables is the partial correlation coefficient. However, for a
general graphical model, we have multi-dimensional potential
functions between two random variables. We hope to find a scalar
quantity to represent the interaction between them as well.

\subsubsection{Walk-summability For Pairwise Binary Graphs}
\cite{NIPS2009_Yusuke} introduced weights on edges of an arbitrary
binary graph, defined an edge zeta function based on those weights
and related the convexity of Bethe free energy with the edge zeta
function. Specifically, given $\mathcal{P}$ be the set of prime
cycles $\{v_{k_0}v_{k_1}...v_{k_{i-1}}v_{k_i}...v_{k_l}v_{k_0}\}$
defined in Definition \ref{def:prime cycle}, for given weights
$\textbf{u}$, the edge zeta function is defined in
\cite{NIPS2009_Yusuke} by
\begin{equation*}
\zeta_{G}(\textbf{u}):=\prod_{w\in
\mathcal{P}}(1-g(w))^{-1},g(w):=u_{v_{k_0}v_{k_1}}...u_{v_{k_{i-1}}v_{k_i}}...u_{v_{k_l}v_{k_0}}.
\end{equation*}
We find that $(1-g(w))^{-1}=\sum_{i=0}^\infty(g(w))^i$, which
represents the walk sums of a prime cycle and its repeated
concatenations.

They introduced an adjacency matrix of directed edges, which is
defined as follows:
\begin{eqnarray*}
\mathcal{M}_{i\rightarrow j,p\rightarrow
q}=\left\{\begin{array}{c}
1,\quad\textrm{if}~p\in\Gamma_{i}\backslash
j,\\0,\quad\textrm{otherwise}.\end{array}\right.
\end{eqnarray*}
Here we use $i\rightarrow j$ rather than $ij$ to explicitly
represent directed edge. They showed that
\begin{equation}
\zeta_{G}(\textbf{u})=det(I-\mathcal{U}\mathcal{M})^{-1}=\prod_{w\in
\mathcal{P}}(1-g(w))^{-1},\label{eq:zeta determinant}
\end{equation}
where $\mathcal{U}$ is a diagonal matrix defined by
$\mathcal{U}_{i\rightarrow j,p\rightarrow q}=u_{i\rightarrow
j}\delta_{i\rightarrow j,p\rightarrow q}$. Let us define two
directed edges $i\rightarrow j$ and $p\rightarrow q$ satisfying
$p\in \Gamma_i\backslash j$ as {\textbf{adjacent}} edges, and call
$\mathcal{U}\mathcal{M}$ an {\textbf{interaction coefficient
matrix}} for adjacent edges. Therefore, Equation \eqref{eq:zeta
determinant} relates summation of weighted prime cycles with
interaction coefficient matrix.

%{Comparatively, for Gaussian graphical model,
%$J^{-1}=(I-R)^{-1}=\sum_{l=0}^{\infty}R^l$ and
%$(R^l)_{ij}=\sum_{w:i\stackrel{l}{\longrightarrow} j}\phi(w)$,
%which characterizes relationship between summation of weighted
%walks and partial correlation coefficient matrix.

\cite{NIPS2009_Yusuke} further defined weights as follows:
\begin{equation*}
u_{i\rightarrow j}:=\frac{\chi_{ij}-m_im_j}{1-m_j^2},
\end{equation*}
where mean $m_i=E_{b_i}[x_i]$ and correlation
$\chi_{ij}=E_{b_{ij}}[x_ix_j]$. Let
$Spec(\mathcal{U}\mathcal{M})\subset \mathbb{C}$ denote the
spectra. They presented the following theorem.
\begin{theorem}(Theorem 4.,\cite{NIPS2009_Yusuke})\\
Given $\mathcal{U}$,$\mathcal{M}$,$m_i$ and $\chi_{ij}$,
$Spec(\mathcal{U}\mathcal{M})$$\subset \mathbb{C}\backslash
\mathbb{R}_{\geq 1}$$\Longrightarrow$ Hessian of Bethe free energy
is positive definite at $\{m_i,\chi_{ij}\}$.
\end{theorem}
We can see that $\prod_{w\in \mathcal{P}}(1-g(w))^{-1}=\prod_{w\in
\mathcal{P}}\sum_{i=0}^\infty(g(w))^i$ in Equation \eqref{eq:zeta
determinant} is well defined, when
$Spec(\mathcal{U}\mathcal{M})$$\subset \mathbb{C}\backslash
\mathbb{R}_{\geq 1}$. Therefore, we can define
$Spec(\mathcal{U}\mathcal{M})$$\subset \mathbb{C}\backslash
\mathbb{R}_{\geq 1}$ as walk-summability condition for pairwise
binary graphs. Since convexity of Bethe free energy implies the
uniqueness of fixed point, the walk-summability condition is
equivalent to convergence condition of LBP.

%We introduce a {\textbf{weight matrix}} $\mathbb{U}$, and
%$\mathbb{U}_{ij}=u_{i\rightarrow j}$. Notice that $\mathbb{U}$ is
%not symmetric.
%$(\mathbb{U}^l)_{qj}=\sum_{w:q\stackrel{l}{\longrightarrow}
%j}g(w)$ corresponds to weighted walks of length $l$ from $q$ to
%$j$, while $((\mathcal{U}\mathcal{M})^l)_{i\rightarrow
%j,p\rightarrow q}$ corresponds to weighted walks of length $l$
%from $q$ of edge $p\rightarrow q$ to $j$ of edge $i\rightarrow j$.
%They are actually related in terms of weighted walks as follows:
%\begin{equation*}
%(\mathbb{U}^l)_{qj}=\sum_{i\in\Gamma_j}((\mathcal{U}\mathcal{M})^l)_{i\rightarrow
%j,p\rightarrow q}+\sum_{i\in\Gamma_j}u_{q\rightarrow
%p}((\mathcal{U}\mathcal{M})^{l-1})_{i\rightarrow j,q\rightarrow
%p}.
%\end{equation*}

Unlike correlation coefficient between two nodes (random
variables), interaction coefficient is between two edges. A
symmetrization of $u_{i\rightarrow j}$ and $u_{j\rightarrow i}$
was defined in \cite{NIPS2009_Yusuke} by
\begin{equation*}
\beta_{ij}:=\frac{\chi_{ij}-m_im_j}{\{(1-m_i^2)(1-m_j^2)\}^{1/2}}=\frac{\textrm{Cov}_{b_{ij}}[x_i,x_j]}{\{\textrm{Var}_{b_i}[x_i]\textrm{Var}_{b_j}[x_j]\}^{1/2}}.
\end{equation*}
$\beta_{ij}$ is the correlation coefficient between $x_i$ and
$x_j$. They showed
$\mathrm{Spec}(\mathcal{U}\mathcal{M})=\mathrm{Spec}(\mathcal{B}\mathcal{M})$,
where $(\mathcal{B})_{i\rightarrow j,p\rightarrow
q}=\beta_{ij}\delta_{i\rightarrow j,p\rightarrow q}$. Therefore,
similar to Gaussian graphical model, for an arbitrary binary
graph, we can also use correlation coefficient $\beta_{ij}$ to
characterize the interaction between two random variables and
analyze the convergence of LBP.

We find another interaction coefficient matrix in
\cite{MooijKappen_IEEETIT_07}. They proved that for pairwise
binary graphs, LBP converges to a unique fixed point, if the
spectral radius of $\mathcal{A}\mathcal{M}$ is strictly smaller
than $1$, where $\mathcal{A}_{i\rightarrow j,p\rightarrow
q}:=tanh|J_{ij}|\delta_{i\rightarrow j,p\rightarrow q}$.
$\mathcal{A}\mathcal{M}$ is also an interaction coefficient matrix
between neighboring edges. We can see
$Spec(\mathcal{B}\mathcal{M})\subset \mathbb{C}\backslash
\mathbb{R}_{\geq 1}$ or $Spec(\mathcal{A}\mathcal{M})\subset
\mathbb{C}\backslash \mathbb{R}_{\geq 1}$ as a walk-summable
condition for binary graphs. However, \cite[Lemma
3]{NIPS2009_Yusuke} showed that: given $\beta_{ij}$ at any fixed
point of LBP, $|\beta_{ij}|\leq tanh|J_{ij}|$. In other words
$\mathcal{B}\mathcal{M}$ is tighter than $\mathcal{A}\mathcal{M}$.

\subsubsection{Walk-summability For General Pairwise Graphs}
In the non-uniform convergence condition in Theorem
\ref{theo:nonuniform convergence}, for a $N$-th level Bethe tree,
we add up all the $N$-th step walks from a root node, where the
weight on edge $(i,j)$ is the quantity
$\frac{d(\psi_{ij})^2-1}{d(\psi_{ij})^2+1}$ and
$d(\psi_{ij})^2=\sup_{x_i,x_j,\hat{x}_i,\hat{x}_j}\sqrt{\frac{\psi_{ij}(x_i,x_j)\psi_{ij}(\hat{x}_i,\hat{x}_j)}{\psi_{ij}(\hat{x}_i,x_j)\psi_{ij}(x_i,\hat{x}_j)}}$.
Let $\mathcal{W}$ be the interaction coefficient matrix with
$(\mathcal{W})_{i\rightarrow j,p\rightarrow
q}=w_{ij}\delta_{i\rightarrow j,p\rightarrow q}$ and
$w_{ij}=\frac{d(\psi_{ij})^2-1}{d(\psi_{ij})^2+1}$. We define the
walk-summability of a general graphical model as follows:
\begin{definition}{\bf (Walk-summability of General Graphical
Model)}\\
Given $\mathcal{W}$, a general pairwise graphical model is
walk-summable, when
$\rho(\mathcal{W}\mathcal{M})<1$.\label{def:walk-summable general}
\end{definition}

Like that for binary graphs, the convergence condition is
equivalent to the walk-summability of a general graphical model
with the interaction coefficient matrix $\mathcal{W}\mathcal{M}$,
which is proved by the following theorem.
\begin{theorem}(Theorem 4.,\cite{MooijKappen_IEEETIT_07})\\
For general pairwise graphical model, LBP converges to a unique
fixed point, when spectral radius
$\rho(\mathcal{W}\mathcal{M})<1$.
\end{theorem}

\begin{lemma}
Our non-uniform convergence condition in Theorem
\ref{theo:nonuniform convergence} is better than Theorem 4
in~\cite{MooijKappen_IEEETIT_07}, or walk-summable condition in
Definition \ref{def:walk-summable general}.\label{lemma:compare
with walk-sum}
\end{lemma}
\begin{proof}
Let $A=\mathcal{W}\mathcal{M}$. $\rho(A)<1$ is equivalent to
$\|A^{N}\|_1<1,N\rightarrow\infty$
(~\cite{MooijKappen_IEEETIT_07}). $(A^{N})_{i\rightarrow
j,k\rightarrow l}$ is the summation of $N$-step weighted walks
from edge $k\rightarrow l$ to $i\rightarrow j$, including
backtracking walks. However, the walk-sum in \eqref{eq:non-uniform
walk sums} for a $N$-level Bethe tree does not include
backtracking walks; thus, it is smaller than $\|A^{N}\|_1$.
Therefore, our non-uniform convergence condition in Theorem
\ref{theo:nonuniform convergence} is milder than $\rho(A)<1$, or
walk-summable condition, which is illustrated in
Fig.~\ref{fig:tightness_nonuniform}(a).
\end{proof}

By ''milder", we mean the set satisfying the sufficient
convergence condition is bigger. Since our non-uniform convergence
condition is derived from \cite[Theorem 14]{ihler05b} and they are
equivalent for infinite Bethe tree, \cite[Theorem 14]{ihler05b} is
better than ~\cite[Theorem 4]{MooijKappen_IEEETIT_07}. When the
convergence condition based on $N$-level Bethe tree is satisfied,
the convergence condition based on infinite Bethe tree must be
satisfied, because the error bounds are guaranteed to decrease
after $N$ iterations of error propagation. Similarly, convergence
condition based on $N$-level Bethe tree is milder than that based
on SAW tree. Therefore, we obtain mildness of convergence
conditions \textit{}, which is shown in
Fig.~\ref{fig:tightness_nonuniform}(b).

In the following, we will analyze the performance of LBP with
respect to accuracy and convergence rate.

\section{Accuracy Bounds for Loopy Belief Propagation}
\label{sec:accuracy}

Recently, \cite{ihler07b} presented an accuracy bound for LBP
which relates the belief of a random variable to its true
marginal. He showed that there exists a configuration on some
nodes of the SAW tree rooted at certain node $s$ of the original
graph, such that the true maginal at node $s$ of the original
graph is equal to the belief at root $s$ of the SAW tree.
Therefore, given certain external force functions on a subset of
nodes, he adopted the non-uniform distance bound in~\cite[Thm.
14]{ihler05b} to obtain an accuracy bound between beliefs and true
marginals.

Given $d(p(x)/b(x))\leq\delta$, his accuracy bound is as follows:
\begin{equation}
\frac{b(x)}{\delta^2+(1-\delta^2)b(x)}\leq
p(x)\leq\frac{\delta^2b(x)}{1-(1-\delta^2)b(x)},\label{eq:Ihler
accuracy bd}
\end{equation}
where $\delta$ is an error bound in dynamic-range measure, $p(x)$
is the normalized true marginal and $b(x)$ is the normalized
belief. Note that $\delta$ in~\cite[Lemma 5]{ihler07b} should be
$\delta^2$.

Because our {\em improved non-uniform distance bound} has been
shown tighter than his non-uniform bound, we can improve his
accuracy bound between the belief and the true marginal. Let
$\max_{x}|\log{p(x)/b(x)}|\leq\log{\varepsilon}$, where
$\varepsilon$ is an error bound in maximum-error measure applying
our Corollary \ref{coro:nonuniform distance}, under certain
external force functions on a subset of nodes of a SAW tree.
Therefore, we have the accuracy bound as $b(x)/\varepsilon\leq
p(x)\leq\varepsilon b(x)$, where $\varepsilon<\delta^2$. Combining
our accuracy bound with the bound in~\eqref{eq:Ihler accuracy bd},
we have the improved bound
\begin{equation*}
\max\{b(x)/\varepsilon,\frac{b(x)}{\delta^2+(1-\delta^2)b(x)}\}\leq
p(x)\leq\min\{\varepsilon
b(x),\frac{\delta^2b(x)}{1-(1-\delta^2)b(x)}\}.
\end{equation*}

\section{Rate of Convergence and Residual Scheduling}
\label{sec:residual scheduling}

For an iterative algorithm such as LBP, the rate of convergence is
an important criteria of performance. We will analyze the
convergence rate of LBP by looking into the gradient of error
bounds on messages. The error bound-variation function
$G_{sp}(\log{\varepsilon})$ in~\eqref{eq:bd variation} is a
measure of the variation of error bounds between successive
iterations; on the other hand, it reflects how fast LBP converges,
because the smaller $G_{sp}(\log{\varepsilon})$ is, the faster
error bounds tighten. Because dynamic-range measure is better than
maximum-error measure in terms of convergence of LBP, we will use
the following error bound-variation function:
\begin{equation*}
G_{sp}(\log{\varepsilon})=\log{\prod_{t\in\Gamma_s\backslash
p}\frac{d(\psi_{ts})^2\varepsilon+1}{d(\psi_{ts})^2+\varepsilon}}-\log{\varepsilon},
\end{equation*}
where $\varepsilon$ is an error bound in dynamic-range measure on
incoming error product. We will use the first derivative of the
function as a metric on the rate of convergence:
\begin{equation*}
G_{sp}^{(1)}(\log{\varepsilon})=\sum_{t\in\Gamma_s\backslash
p}\frac{\varepsilon((d(\psi_{ts})^4-1)}{(d(\psi_{ts})^2\varepsilon+1)(d(\psi_{ts})^2+\varepsilon)}-1.
\end{equation*}
Recall that $G_{sp}^{(1)}(\log{\varepsilon})$ should be less than
zero to ensure convergence. When we have infinitesimal error
disturbance, $|G_{sp}^{(1)}(0)|$ will be used as a local rate of
convergence. Because our rate of convergence varies on each
direction of message passing, messages on the direction with the
greatest rate will be updated prior to others in dynamic
scheduling.

Some works have been done to utilize message residuals as a way of
priority in dynamic scheduling by~\cite{Elidan06}
and~\cite{Sutton07}. Rather than calculating future message
residuals,~\cite{Sutton07} utilized their upper-bounds as
estimates of message residuals in their scheduling algorithm {\em
RBP0L}. They adopted maximum-error measure as a metric of message
residuals, which was defined by them as
$r(m_{ts})=\max_{x_{s}}|\log{e_{ts}(x_s)}|$. They showed that by
the contraction property of maximum-error measure it can be
upper-bounded as $r(m_{ts})\leq\sum_{u\in\Gamma_{t}\backslash
s}r(m_{ut})$. However, their upper-bound is not theoretically
sound, because they ignored the normalization factor in their
proof. Therefore, we can modify their {\em RBP0L} by utilizing our
upper-bound in \eqref{eq:maxcontraction}.

\section{Fixed Points and Message Errors for Uniform Binary Graphs}
\label{sec:uniform graph}

\cite{MooijKappen_JSTAT_05} analyzed the phase transition for
binary graphs based on Hessian of Bethe free energy. They
presented ferromagnetic interactions, antiferromagnetic
interactions and spin-glass interactions, by analyzing stability
of paramagnetic fixed point and other stable or unstable fixed
points. \cite{NIPS2009_Yusuke} obtained several interesting
results on binary graphs based on edge zeta function and Bethe
free energy. They stated that Bethe free energy is never convex
for any connected graph with at least two linearly independent
cycles. They also stated that the number of the fixed points of
LBP is always odd for binary graphs. We will analyze the behavior
of fixed points of LBP based on message updating function
directly.

In Section~\ref{sec:bounds belief}, we discussed uniform and
non-uniform distance bounds on beliefs. An error bound-variation
function was introduced to study the variation of error bounds
between successive iterations. However, to study the mechanism
behind message passing, we are more interested to know the
variation of true errors. Since it is usually hard to formulate
the true error-variation function for general graphical models, in
this section, we will only explore true error variation functions
for binary graphs.

Let us first introduce a well-studied binary graph -- Ising model.
The probability measure of Ising model can be expressed as:
\begin{equation}
P(x)=\frac{1}{Z}\exp{(\sum_{(s,t)\in\mathbb{E}}J_{st}x_sx_t+\sum_{s\in\mathbb{V}}
\theta_s x_s)},\label{eq:Ising}
\end{equation}
corresponding to $\psi_{st}(x_s,x_t)=\exp{(J_{st}x_sx_t)}$ and
$\psi_s(x_s)=\exp{(\theta_s x_s)}$ in \eqref{eq:pdf}. Because
$\{x_s\}$ are $\pm1$-valued, potential functions can also be
expressed as
$\begin{pmatrix}\exp{(J_{st})}&\exp{(-J_{st})}\\\exp{(-J_{st})}&\exp{(J_{st})}\end{pmatrix}$
and
$\begin{pmatrix}\exp{(\theta_s)}\\\exp{(-\theta_s)}\end{pmatrix}$.
However, rather than working on the Ising model, we will study a
more simple model. We call it completely uniform model (uniform
connectivity, uniform potential functions), which has the pairwise
potential functions $\begin{pmatrix}a&b\\b&a\end{pmatrix}$ and
single-node potential functions
$\begin{pmatrix}c\\d\end{pmatrix}$, where $a,b,c,d$ are positive.
Similar to~\eqref{eq:singlebelive}, we will put single-node
potential functions into beliefs and only discuss the influence of
pairwise potential functions on message errors. We can easily find
that a completely uniform graph has uniform messages.
\begin{property}
For a completely uniform graphical model, when synchronous LBP
reaches a steady state, all messages are the
same.\label{prop:steady state}
\end{property}
\begin{proof}
Completely uniform graphs are topologically invariant for each
node. In other words, each message has the same LBP update
equation. If some messages are different, for the symmetric
network, LBP will not reach a steady state.
\end{proof}
Because all messages have the same LBP update equation, we can
calculate the fixed-point messages exactly and discuss the
distances between them.

\subsection{Fixed Points and Quasi-Fixed Points}
Let us first discuss fixed-point messages for completely uniform
graphs. Assume the degree of each node is $k$. Let
$m_{out}=\begin{pmatrix} y\\1-y
\end{pmatrix}$ denote the outgoing message and
$m_{in}=\begin{pmatrix} x\\1-x \end{pmatrix}$ denote each incoming
message.  Therefore, we have the following LBP updating function:
\begin{equation}
y=F(x)=\frac{ax^k+b(1-x)^k}{(a+b)(x^k+(1-x)^k)}.\label{eq:binary
message}
\end{equation}
We can easily find that~\eqref{eq:binary message} is symmetric
with respect to the point $(x=0.5,y=0.5)$. Synchronous LBP update
corresponds to the fixed-point iteration function
$x_{n+1}=F(x_n)$, where $n$ is the iteration number. When
$x_{n+1}=x_{n}$, LBP message reaches a {\em fixed point}. However,
we sometimes have $x_{n+k}=x_{n}$ or $F^k(x)=x$, where $F^k(x)$ is
the composition function of $F(x)$ with itself $k$ times, which
shows $k$th-order periodicity. We define the solutions to
$F^k(x)=x,k>1$ as {\em quasi-fixed points}, when a belief network
will oscillate. In the following, we will show that LBP for
completely uniform binary graphs will have at most second order
periodicity.

\begin{property}
LBP updating function in~\eqref{eq:binary message} has at most
three real fixed points.\label{prop:three fixed points}
\end{property}
\begin{proof}
The second derivative of $F(x)$ is as follows: when $a>b$
\begin{equation*}
F^{(2)}(x)=((2x-k-1)x^k+(2x+k-1)(1-x)^k)\times\frac{k(a-b)x^{k-2}(1-x)^{k-2}}{(a+b)(x^k+(1-x)^k)^3}=\left\{\begin{array}{c}
>0,x\in(0,0.5)\\<0,x\in(0.5,1)\\=0,x=0,1,0.5\end{array}\right..
\end{equation*}
We can see that $F(x)$ is strictly convex when $0<x<0.5$ and
strictly concave when $0.5<x<1$. Similarly, for $a<b$, $F(x)$ is
strictly concave when $0<x<0.5$ and strictly convex when
$0.5<x<1$. When this function intersects with an arbitrary line,
there must be at most three crossing points. As shown in
Fig.~\ref{fig:binary LBP}(a), it must have at most three crossings
with $y=x$; similarly with $y=1-x$ in Fig.~\ref{fig:binary
LBP}(b).
\end{proof}
\begin{figure}[htb]
\begin{center}
\includegraphics[height=5.14 cm,width=10.94 cm]{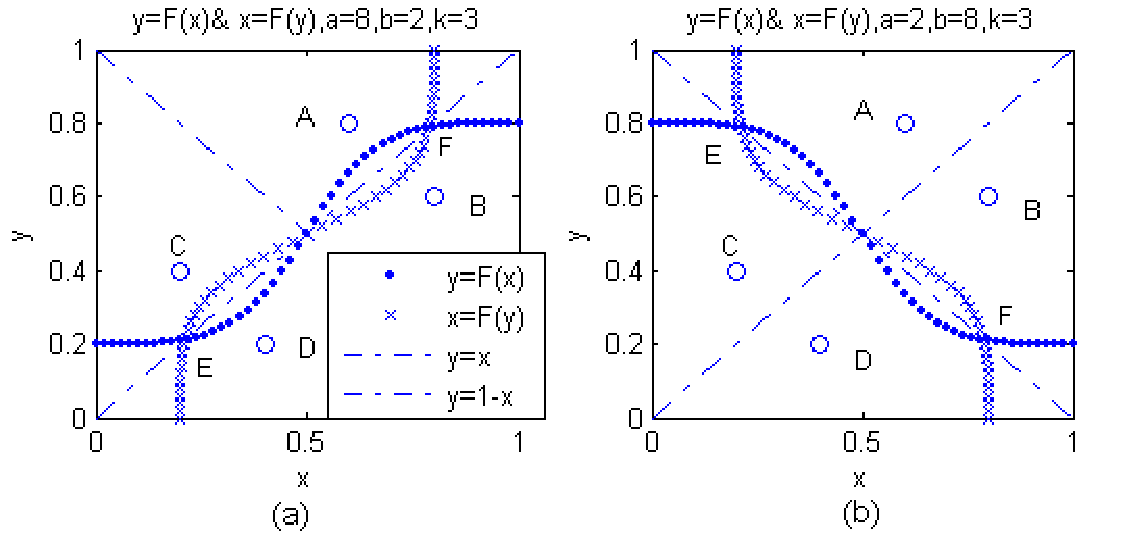}
\end{center}
\caption{LBP updating function in \eqref{eq:binary message} for
$a>b$ and $a<b$.}\label{fig:binary LBP}
\end{figure}
This property conforms to the analysis of
\cite{MooijKappen_JSTAT_05} and \cite{NIPS2009_Yusuke}. We will
show the symmetry of fixed-point messages for uniform binary
graphs as follows.
\begin{property}
For a completely uniform binary graph, synchronous LBP will either
converge to the unique fixed point
$\begin{pmatrix}0.5\\0.5\end{pmatrix}$ (paramagnetic fixed point),
or converge to one of $\begin{pmatrix}x^*\\1-x^*\end{pmatrix}$ and
$\begin{pmatrix}1-x^*\\x^*\end{pmatrix}$ when $a>b$
(ferromagnetic), or oscillate between
$\begin{pmatrix}x^*\\1-x^*\end{pmatrix}$ and
$\begin{pmatrix}1-x^*\\x^*\end{pmatrix}$ when $a<b$
(anti-ferromagnetic). When $a>b$, $x^*$ is the solution to
$x^*=F(x^*)$; otherwise, $x^*$ is the solution to
$1-x^*=F(x^*)$.\label{prop:fixed or oscilloate}
\end{property}

The proof appears in Appendix A.

From the previous property, we can conclude that completely
uniform binary graphs will have at most second order periodicity.
In other words, $F^{2n}(x)=x\Leftrightarrow F^{2}(x)=x$ and
$F^{2n-1}(x)=x\Leftrightarrow F(x)=x$.

Let us calculate the fixed points and quasi-fixed points for the
uniform graph in Fig.~\ref{fig:graphs}(c) with $a=\eta$ and
$b=1-\eta$. Solving $x=\frac{\eta
x^3+(1-\eta)(1-x)^3}{x^3+(1-x)^3}$ and $1-x=\frac{\eta
x^3+(1-\eta)(1-x)^3}{x^3+(1-x)^3}$ yields the fixed points and
quasi-fixed points respectively, for the graph in
Fig.~\ref{fig:graphs}(c).  Specifically, we can obtain four
solutions of fixed points
$\{\frac{1}{2},\frac{1}{2},\frac{-2+\eta-\sqrt{-4+8\eta-3\eta^2}}{2(-2+\eta)},
\frac{-2+\eta+\sqrt{-4+8\eta-3\eta^2}}{2(-2+\eta)}\}$ and four
solutions of quasi-fixed points
$\{\frac{1}{2},\frac{1}{2},\frac{1+\eta-\sqrt{1-2\eta-3\eta^2}}{2(1+\eta)},
\frac{1+\eta+\sqrt{1-2\eta-3\eta^2}}{2(1+\eta)}\}$.  When
$\eta>2/3$, the graph has two real fixed points except $0.5$; when
$\eta<1/3$, the graph has two real quasi-fixed points except
$0.5$; when $1/3<\eta<2/3$, the graph has one real fixed point
$0.5$. For instance, when $\eta=0.7$, we have two stable fixed
points $(0.9071, 0.0929)$ and $(0.0929, 0.9071)$; when $\eta=0.3$,
we have two quasi-fixed points $(0.9071, 0.0929)$ and $(0.0929,
0.9071)$. We observe that both cases have the same strength of
potential function $d(\psi)^2=0.7/0.3$, though their dynamic
characteristics are different.

Based on Property \ref{prop:fixed or oscilloate}, we find that for
completely uniform graphs, the maximum multiplicative error and
the minimum multiplicative error between two fixed-point messages
are reciprocal. In other words, $d(e(x))=\max{e(x)}$. Therefore,
compared to our {\em uniform distance bound} in Corollary
\ref{coro:uniform distance bd}, we have a tighter distance bound
as follows.
\begin{corollary}{\textbf{ (Uniform Distance Bound for Completely Uniform Binary Graph)}}\\
$\mathbb{G}(\mathbb{V},\mathbb{E})$ is a completely uniform binary
graphical model. The log-distance bound on beliefs at node $s$ is
\[\sum_{t\in\Gamma_s}\log\frac{d(\psi_{ts})^2
\varepsilon+1}{d(\psi_{ts})^2+\varepsilon},\] where $\varepsilon$
should satisfy \[\log{\varepsilon}=\max_{(s,p)\in
\mathbb{E}}\sum_{t\in\Gamma_s\backslash
p}\log{\frac{d(\psi_{ts})^2\varepsilon+1}{d(\psi_{ts})^2+\varepsilon}}.\]\label{coro:uniform
distance bd completely uniform graph}
\end{corollary}
\begin{proof}
$\log\max E_s=\log d(E_s)\leq\sum_{t\in\Gamma_s}\log d(e_{ts})\leq
\sum_{t\in\Gamma_s}\log\frac{d(\psi_{ts})^2\varepsilon+1}{d(\psi_{ts})^2+\varepsilon}$.
\end{proof}
For the uniform graph in Fig.~\ref{fig:graphs}(c), when
$\eta=0.7$, we have the true log-distance equal to $2.2785$, while
our previous log-distance bound in Corollary \ref{coro:uniform
distance bd completely uniform graph} obtains $2.2785$, which is
exactly equal to the true value, and our {\em Improved-UDB} in
Corollary 7 obtains $2.3318$.

\subsection{True Error-Variation Function}
\begin{figure}[htb]
\begin{center}
\includegraphics[height=8.18 cm,width=11.09 cm]{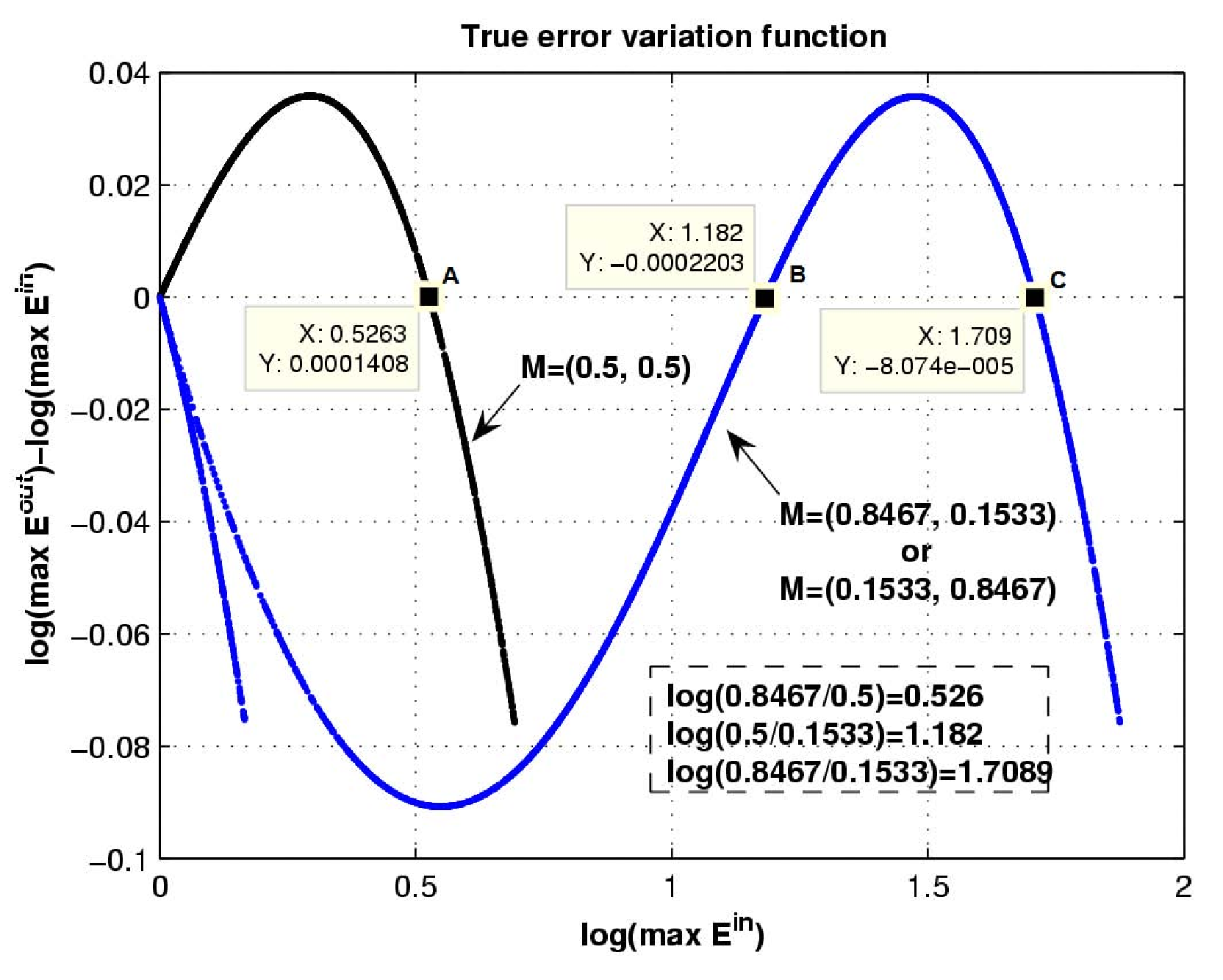}
\end{center}
\caption{True error variation function when $M$s are fixed-point
messages for the completely uniform graph in
Fig.~\ref{fig:graphs}(c). $a=0.7$, $b=0.3$. The fixed-point
messages are: $M=(0.8467,0.1533)$, $M=(0.1533,0.8467)$ and
$M=(0.5,0.5)$.}\label{fig:true variation}
\end{figure}

In this section, we characterize the true error-variation function
for a completely uniform binary graph. We have the following
message updating equation:
\begin{equation*}
\begin{pmatrix} m e^{out}_1\\(1-m)e^{out}_2\end{pmatrix}=\frac{1}{a+b}\begin{pmatrix}
a&b\\b&a\end{pmatrix}\begin{pmatrix}M
E^{in}_1\\(1-M)E^{in}_2\end{pmatrix},
\end{equation*}
where $M$ is the product of fixed-point incoming messages, $m$ is
the fixed-point outgoing message, $E^{in}$ represents the product
of incoming errors and $e^{out}$ represents the outgoing error.
Assuming $E^{in}$ is the same for each node at a level on the
Bethe tree, we have the following error equation:
\begin{equation*}
\begin{pmatrix}E^{out}_1\\E^{out}_2\end{pmatrix}=\frac{(aM+b(1-M))^k+(bM+a(1-M))^k}{(aME^{in}_1+b(1-M)E^{in}_2)^k+(bME^{in}_1+a(1-M)E^{in}_2)^k}\begin{pmatrix}\frac{(aME^{in}_1+b(1-M)E^{in}_2)^k}{(aM+b(1-M))^k}\\\frac{(bME^{in}_1+a(1-M)E^{in}_2)^k}{(bM+a(1-M))^k}\end{pmatrix},
\end{equation*}
where $E^{out}$ is the product of outgoing errors flowing into a
node at the upper level.

When $E^{in}_1>E^{in}_2$ and $a>b$, we have $E^{out}_1>E^{out}_2$.
Therefore, letting $E$ denote $E^{in}_1$, we obtain the true error
variation function:
\begin{equation*}
G(log(E))=\log\max{E^{out}}-\log\max{E^{in}}
\end{equation*}
\begin{equation}
=\log{(\frac{(aME+b(1-ME))^k}{(aM+b(1-M))^k}\cdot\frac{(aM+b(1-M))^k+(bM+a(1-M))^k}{(aME+b(1-ME))^k+(bME+a(1-ME))^k})}-\log{E},\label{eq:true
error bd}
\end{equation}
when $1<E<1/M$ and $a>b$.

An example of the true error variation function is illustrated in
Fig.~\ref{fig:true variation} for the graphical model in
Fig.~\ref{fig:graphs} (c). The curve of the error variation
function $G(\log{E})$ in Equation \eqref{eq:true error bd} varies
with the choice of $M$. The black curve corresponds to
$G(\log{E})$ for $M = (0.5, 0.5)$, while the blue curve
corresponds to $G(\log{E})$ for $M = (0.8467, 0.1533)$ or $M =
(0.1533, 0.8467)$. Since $G^{(1)}(\infty) = -1$, when $G(\log{E})$
does not cross the horizontal axis except the point at
$\log{E}=0$, we have $G(\log{E})<0$ for $\log{E}>0$. In other
words, $\log{E}$ will eventually decrease to zero and LBP
converges to a unique fixed point. However, when $G(\log{E})$
crosses the horizontal axis besides $\log{E}=0$, $\log{E}$ will
eventually stay at stable points, in which case, the product of
the incoming errors at one level of Bethe tree equals the product
of the incoming errors at its upper level. In other words, errors
will not decrease after one LBP update. In Fig.~\ref{fig:true
variation}, for the black curve, when $\log{E}$ leaves zero, it
will eventually stay at $A$. For the blue curve in
Fig.~\ref{fig:true variation}, when $\log{E}$ is between zero and
the value at point $B$, it will decrease and finally stay at zero;
when $\log{E}$ is bigger than the value at point $B$, it will
increase and finally stay at point $C$. We can see that point $B$
is an unstable point.

From the example in Fig.~\ref{fig:true variation}, we can observe
that the zero-crossing points of $\log{E}$ correspond to the exact
log distances between two fixed-point messages. Specifically, the
value at point $A$ is equal to the maximal log distance between $M
= (0.8467, 0.1533)$ and $M = (0.5, 0.5)$, and the value at point
$B$ is equal to the maximal log distance between $M = (0.5, 0.5)$
and $M = (0.1533, 0.8467)$, and the value at point $C$ is equal to
the maximal log distance between $M = (0.8467, 0.1533)$ and $M =
(0.1533, 0.8467)$. Therefore, our true error function in
Equation~\eqref{eq:true error bd} characterizes the true distance
between fixed points, when LBP does not converge.

\section{Conclusion}
\label{sec:conclusion}

In this paper, we presented tighter error bounds on Loopy Belief
Propagation (LBP) and used these bounds to study the
dynamics---error, convergence, accuracy, and scheduling---of the
sum-product algorithm. Specifically, we derived tight upper- and
lower-bounds on error propagation in synchronous belief networks.
We subsequently relied on these bounds to provide uniform and
non-uniform distance bounds for the sum-product algorithm. We then
used the distance bounds to obtain uniform and non-uniform
sufficient conditions for convergence of the sum-product
algorithm. We investigated the relation between convergence of LBP
with sparsity and walk-summability of graphical models. We also
showed that upper-bounds on message errors can be utilized to
determine a priority for scheduling in sequential belief
propagation. Moreover, we studied the accuracy of the bounds on
the sum-product algorithm based on our error bounds. We also
presented a case study of LBP by characterizing the dynamics of
the sum-product algorithm for completely uniform graphs and
analyzed its fixed and quasi-fixed (oscillatory) points.

%{\noindent \em Remainder omitted in this sample. See http://www.jmlr.org/papers/ for full paper.}

% Acknowledgements should go at the end, before appendices and references

%\acks{We would like to acknowledge support for this project
%from the National Science Foundation (NSF grant IIS-9988642)
%and the Multidisciplinary Research Program of the Department
%of Defense (MURI N00014-00-1-0637). }


\newpage

%\appendix
\section*{Appendix A. Detailed Proofs}
\label{app:theorem}

\textbf{Proof of Theorem \ref{theo:error upper&lower}}
\begin{proof}
We use {\em maximum multiplicative error} function as an error
measure:
\begin{equation*}
\max_{x_s} e^{i+1}_{ts}(x_{s})=\max_{x_s}
   \frac{\int\psi_{ts}(x_t,x_s)M_{ts}(x_t)E^i_{ts}(x_t)dx_t}{\int\psi_{t\star}(x_t)M_{ts}(x_t)E^i_{ts}(x_t)dx_t}\times
\frac{\int\psi_{t\star}(x_t)M_{ts}(x_t)dx_t}{\int\psi_{ts}(x_t,x_s)M_{ts}(x_t)dx_t},
\end{equation*}
where $\psi_{t\star}(x_t) = \int\psi_{ts}(x_t,x_s)dx_s$. The {\em
minimum multiplicative error} function $\min_{x_s}
e^{i+1}_{ts}(x_{s})$ is also used as an error measure in this
theorem. Some assumptions throughout this proof are:
$\psi_{ts}(x_t,x_s)$ is positive; message product $M_{ts}(x_t)$
and polluted message product $M_{ts}(x_t)E^i_{ts}(x_t)$ are
positive and normalized.

We use the same framework of proof as that in \cite[Thm.
8]{ihler05b}. Let us first introduce a lemma that will be used in
our proof.
\begin{lemma}
For $f_1,f_2,g_1,g_2$ all positive,
\begin{eqnarray*}
\frac{f_1+f_2}{g_1+g_2}\leq\max[\frac{f_1}{g_1},\frac{f_2}{g_2}],\quad\frac{f_1+f_2}{g_1+g_2}\geq\min[\frac{f_1}{g_1},\frac{f_2}{g_2}].
\end{eqnarray*}\label{lemma:sum fraction upper lower}
\end{lemma}
\begin{proof}
The left inequality is proved in \cite{ihler05b}. Let us restate
it here. Assume without loss of generality that $f_1/g_1\geq
f_2/g_2$ so that $f_1 g_2\geq f_2 g_1\Rightarrow f_1 g_2+f_1
g_1\geq f_2 g_1+f_1 g_1\Rightarrow
\frac{f_1}{g_1}\geq\frac{f_1+f_2}{g_1+g_2}$. For the right
inequality assume without loss of generality that $f_1/g_1\leq
f_2/g_2$ so that $f_1 g_2\leq f_2 g_1\Rightarrow f_1 g_2+f_1
g_1\leq f_2 g_1+f_1 g_1\Rightarrow
\frac{f_1}{g_1}\leq\frac{f_1+f_2}{g_1+g_2}$.
\end{proof}

Similar to the analysis in~\cite[Lemma 26]{ihler05b},we need the
following lemma to assist our proof. In the following, we shall
omit reference to the iteration number of the messages and errors
for simplicity and clarity of the presentation.

\begin{lemma}
The maximum of $\max_{x_s} e_{ts}(x_{s})$ or the minimum of
$\min_{x_s} e_{ts}(x_{s})$ is attained when
\begin{align*}
&\psi_{ts}(x_t,x_s)=1+(d(\psi_{ts})^2-1)\chi_\psi(x_t),
\quad\psi_{t\star}(x_t)=1+(d(\psi_{t\star})^2-1)\chi_\star(x_t)\\
&E_{ts}(x_t)=1+(d(E_{ts})^2-1)\chi_E(x_t),
\end{align*}
where $\chi_\psi$, $\chi_\star$ and $\chi_E$ are indicator
functions.\label{lemma:max attained at extreme psi}
\end{lemma}
\begin{proof}
Let
$\psi_{ts}(x_t,x_s)=\alpha_1\psi_{1}(x_t,x_s)+\alpha_2\psi_{2}(x_t,x_s)$,where
$\alpha_1\geq0$,$\alpha_2\geq0$, $\alpha_1+\alpha_2=1$. In other
words,$\psi_{ts}(x_t,x_s)$ is a convex combination of two
arbitrary positive functions $\psi_{1}(x_t,x_s)$ and
$\psi_{2}(x_t,x_s)$. Thus, by applying Lemma \ref{lemma:sum
fraction upper lower}, we have:
\begin{eqnarray*}
\frac{\alpha_1\int\psi_{1}(x_t,x_s)M_{ts}(x_t)E_{ts}(x_t)dx_t+\alpha_2\int\psi_{2}(x_t,x_s)M_{ts}(x_t)E_{ts}(x_t)dx_t}{\alpha_1\int\psi_{1}(x_t,x_s)M_{ts}(x_t)dx_t+\alpha_2\int\psi_{2}(x_t,x_s)M_{ts}(x_t)dx_t}\\
\leq\max[\frac{\int\psi_{1}(x_t,x_s)M_{ts}(x_t)E_{ts}(x_t)dx_t}{\int\psi_{1}(x_t,x_s)M_{ts}(x_t)dx_t},\frac{\int\psi_{2}(x_t,x_s)M_{ts}(x_t)E_{ts}(x_t)dx_t}{\int\psi_{2}(x_t,x_s)M_{ts}(x_t)dx_t}].
\end{eqnarray*}
We find that $\max_{x_s} e_{ts}(x_{s})$ is maximized when we take
the maximum of the RHS expression in the previous inequality. Let
us scale $\psi_{ts}(x_t,x_s)$ so that the minimal value of the
function is $1$. Thus, $\psi_{ts}(x_t,x_s)$ can be composed by a
convex combination of functions which have the form
$1+(d(\psi_{ts})^2-1)\chi_\psi(x_t)$, where $\chi_\psi(x_t)$ is an
indicator function. We can find that the $\max_{x_s}
e_{ts}(x_{s})$ is maximized when
$\psi_{ts}(x_t,x_s)=1+(d(\psi_{ts})^2-1)\chi_\psi(x_t)$. Similar
are the proofs for $\psi_{t\star}(x_t)$ and $E_{ts}(x_t)$.

To minimize the $\min_{x_s} e_{ts}(x_{s})$, by applying Lemma
\ref{lemma:sum fraction upper lower}, we have:
\begin{eqnarray*}
\frac{\alpha_1\int\psi_{1}(x_t,x_s)M_{ts}(x_t)E_{ts}(x_t)dx_t+\alpha_2\int\psi_{2}(x_t,x_s)M_{ts}(x_t)E_{ts}(x_t)dx_t}{\alpha_1\int\psi_{1}(x_t,x_s)M_{ts}(x_t)dx_t+\alpha_2\int\psi_{2}(x_t,x_s)M_{ts}(x_t)dx_t}\\
\geq\min[\frac{\int\psi_{1}(x_t,x_s)M_{ts}(x_t)E_{ts}(x_t)dx_t}{\int\psi_{1}(x_t,x_s)M_{ts}(x_t)dx_t},\frac{\int\psi_{2}(x_t,x_s)M_{ts}(x_t)E_{ts}(x_t)dx_t}{\int\psi_{2}(x_t,x_s)M_{ts}(x_t)dx_t}].
\end{eqnarray*}
Furthermore, by constructing the potential function
$\psi_{ts}(x_t,x_s)$ as a convex combination of functions of the
form $1+(d(\psi_{ts})^2-1)\chi_\psi(x_t)$, where $\chi_\psi(x_t)$
is an indicator function, we can find that $\min_{x_s}
e_{ts}(x_{s})$ is minimized when $\psi_{ts}(x_t,x_s))$ is one of
these functions. Similar are the proofs for $\psi_{t\star}(x_t)$
and $E_{ts}(x_t)$.
\end{proof}
%taking values in $\{0,1\}$.
So we have $\max_{x_s} e_{ts}(x_{s})$ is bounded by
\begin{eqnarray*}
&&\frac{\int\psi_{ts}(x_t,x_s)M_{ts}(x_t)E_{ts}(x_t)dx_t}{\int\psi_{t\star}(x_t)M_{ts}(x_t)E_{ts}(x_t)dx_t}\times
\frac{\int\psi_{t\star}(x_t)M_{ts}(x_t)dx_t}{\int\psi_{ts}(x_t,x_s)M_{ts}(x_t)dx_t}\\
&&=\frac{\int(1+(d(\psi_{ts})^2-1)\chi_\psi(x_t))M_{ts}(x_t)(1+(d(E_{ts})^2-1)\chi_E(x_t))dx_t}{\int(1+(d(\psi_{t\star})^2-1)\chi_\star(x_t))M_{ts}(x_t)(1+(d(E_{ts})^2-1)\chi_E(x_t))dx_t}\\
&&\times\frac{\int(1+(d(\psi_{t\star})^2-1)\chi_\star(x_t))M_{ts}(x_t)dx_t}{\int(1+(d(\psi_{ts})^2-1)\chi_\psi(x_t))M_{ts}(x_t)dx_t}.
\end{eqnarray*}
Define the quantities:
\begin{align*}
&M_A =\int M_{ts}(x_t)\chi_\psi(x_t)d x_t,\quad M_B =\int
M_{ts}(x_t)\chi_\star(x_t)d x_t,\quad M_E =\int
M_{ts}(x_t)\chi_E(x_t)d
x_t,\\
&M_{AE} =\int M_{ts}(x_t)\chi_\psi(x_t)\chi_E(x_t)d x_t,\quad M_{BE} =\int M_{ts}(x_t)\chi_\star(x_t)\chi_E(x_t)dx_t,\\
&\alpha_1 =d(\psi_{ts})^2-1,\quad
\alpha_2=d(\psi_{t\star})^2-1,\quad \beta=d(E_{ts})^2-1.
\end{align*}
The maximum multiplicative error $\max_{x_s} e_{ts}(x_s)$ is
upper-bounded by $\Delta_1$ where
\[
\Delta_1 = \max_{M}\frac{1+\alpha_1 M_A+\beta M_E+\alpha_1\beta
M_{AE}}{1+\alpha_2 M_B+\beta M_E+\alpha_2\beta
M_{BE}}\,\frac{1+\alpha_2 M_B}{1+\alpha_1 M_A}.
\]
The maximum is obtained when $M_{AE}=M_A=M_E=1-M_B$ and
$M_{BE}=0$, which gives
\[
%\max_{x_s}e_{ts}(x_s)\leq
\Delta_1 =
\max_{M_E}\frac{1+(\alpha_1+\beta+\alpha_1\beta)M_E}{1+\alpha_2+(\beta-\alpha_2)M_E}
\,\frac{1+\alpha_2-\alpha_2M_E}{1+\alpha_1 M_E}.
\]
Taking the derivative wrt $M_E$ and setting it to zero, we obtain
%that the optimal $M_E$ is $M_E=\frac{\sqrt{1+\alpha_2}}{\sqrt{1+\alpha_1}\sqrt{1+\beta}+\sqrt{1+\alpha_2}}$. Thus, we have the upper-bound
\begin{equation*}
\max_{x_s} e_{ts}(x_s)\leq \Delta_1 =
\left(\frac{d(\psi_{ts})d(\psi_{t\star})d(E_{ts})+1}{d(\psi_{ts})d(\psi_{t\star})+d(E_{ts})}\right)^2.
\end{equation*}

Similarly to what we have done so far, we can lower-bound
$\min_{x_s} e_{ts}(x_{s})$ with respect to $\psi_{ts}(x_t,x_s)$,
$\psi_{t\star}(x_t)$ and $E_{ts}(x_t)$, to obtain
\begin{equation*}
\min_{x_s}
e_{ts}(x_s)\geq\left(\frac{d(\psi_{ts})d(\psi_{t\star})+d(E_{ts})}{d(\psi_{ts})d(\psi_{t\star})d(E_{ts})+1}\right)^2=\frac{1}{\Delta_1}.
\end{equation*}
\end{proof}
\textbf{Proof of Corollary \ref{coro:uniform distance bd}}
\begin{proof}
Let
$\Delta_{ut}(x)=(\frac{d(\psi_{ut})d(\psi_{u\star})x+1}{d(\psi_{ut})d(\psi_{u\star})+x})^2,x\geq1,ut\in\mathbb{E}$.
Therefore,
\begin{align*}
d(E_{ts}^i)\leq\prod_{u\in\Gamma_t\backslash
s}d(e_{ut}^i)=\prod_{u\in\Gamma_t\backslash
s}\frac{\max\sqrt{e_{ut}^i(x_t)}}{\min\sqrt{e_{ut}^i(x_t)}}\leq\varepsilon_{ts}^i=\prod_{u\in\Gamma_t\backslash
s}\Delta_{ut}(d(E_{ut}^{i-1})).
\end{align*}
Thus, we have
\begin{align*}
&\max_{x_s}E_{sp}^{i+1}(x_s)\leq\prod_{t\in\Gamma_s\backslash
p}\max_{x_s}e_{ts}^{i+1}(x_s)\leq\varepsilon_{sp}^{i+1}=\prod_{t\in\Gamma_s\backslash
p}\Delta_{ts}(d(E_{ts}^i))\\
&\leq\prod_{t\in\Gamma_s\backslash
p}\Delta_{ts}(\varepsilon_{ts}^{i})\leq\prod_{t\in\Gamma_s\backslash
p}\Delta_{ts}(\max_{t\in\Gamma_s\backslash
p}\varepsilon_{ts}^{i})=\Delta_3(\max_{t\in\Gamma_s\backslash
p}\varepsilon_{ts}^{i}).
\end{align*}

The term $\varepsilon_{sp}^{i+1}$ is an upper-bound on the
incoming error product $E_{sp}^{i+1}(x_s)$ at iteration $i+1$,
while $\max_{t\in\Gamma_s\backslash p}\varepsilon_{ts}^{i}$ is the
maximum of the upper-bounds on the incoming error products
$\{E_{ts}^{i}(x_t),t\in\Gamma_s\backslash p\}$ at iteration $i$.
We hope to achieve that
$\varepsilon_{sp}^{i+1}<\max_{t\in\Gamma_s\backslash
p}\varepsilon_{ts}^{i}$. Denoting
$\varepsilon=\max_{t\in\Gamma_s\backslash p}\varepsilon_{ts}^{i}$,
let us introduce an {\em error bound-variation function}:
\begin{equation*}
G_{sp}(\log{\varepsilon})=\log
\Delta_3(\varepsilon)-\log{\varepsilon}\geq\log\varepsilon_{sp}^{i+1}-\log\max_{t\in\Gamma_s\backslash
p}\varepsilon_{ts}^{i},\varepsilon\geq 1,
\end{equation*}
which describes variation of error bound after each iteration.
When $G_{sp}(\log{\varepsilon})=0$, the log-distance bound
$\log{\varepsilon}$ will reach a fixed point, which is the maximal
distance between message products at various iterations. Because
$G_{sp}^{(2)}(\log{\varepsilon})<0$ for $\log{\varepsilon}>0$ and
$G_{sp}^{(1)}(\infty)=-1/2$, $G_{sp}^{(1)}(\log{\varepsilon})$
will decrease until it is equal to $-1/2$.  Therefore, it only has
one crossing point besides $\log{\varepsilon}=0$ (zero crossing
point). This nonzero crossing point is a stable fixed point of
function $G_{sp}(\log{\varepsilon})$.  In other words, once
$\log{\varepsilon}$ leaves the zero crossing point, it will stay
at this stable crossing point, $\log{\varepsilon^*}$, which
corresponds to the upper bound on error products.

Because the distance between fixed points of $B_s(x_s)$ is
\[
\log{E_s(x_s)}=\log{\prod_{t\in\Gamma_s}e_{ts}(x_s)}\leq
\log{\prod_{t\in\Gamma_s}\Delta_{ts}(\varepsilon^*)},
\]
we can obtain the log-distance bound on $B_s(x_s)$ by taking the
maximum $\varepsilon^*$.
\end{proof}
\textbf{Proof of Theorem \ref{theo:uniform converge}}
\begin{proof}
Let us revisit the {\em error bound-variation function} in
Equation~\eqref{eq:bd variation}:
\begin{equation*}
G_{sp}(\log{\varepsilon})=\log{\prod_{t\in\Gamma_s\backslash
p}(\frac{d(\psi_{ts})d(\psi_{t\star})\varepsilon+1}{d(\psi_{ts})d(\psi_{t\star})+\varepsilon}})^2-\log{\varepsilon},
\end{equation*}
which describes the variation of the error bound after each
iteration. To guarantee that LBP converges, it is sufficient to
require $G_{sp}(\log{\varepsilon})<0,\forall \log{\varepsilon}>0$.
Let $z=\log{\varepsilon}$. The second derivative of $G_{sp}(z)$ is
\begin{equation*}
G_{sp}^{(2)}(z)=2\times\sum_{t\in\Gamma_s\backslash
p}\frac{d(\psi_{ts})d(\psi_{t\star})e^z((d(\psi_{ts})d(\psi_{t\star}))^2-1)(1-e^{2z})}{(d(\psi_{ts})d(\psi_{t\star})e^z+1)^2(d(\psi_{ts})d(\psi_{t\star})+e^z)^2}\leq
0,
\end{equation*}
when $d(\psi_{ts})d(\psi_{t\star})>1$ and $z\geq 0$. When $z>0$,
$G_{sp}(z)$ is strictly concave.

The first derivation of $G_{sp}(z)$ is
\begin{equation*}
G_{sp}^{(1)}(z)=2\times\sum_{t\in\Gamma_s\backslash
p}\frac{e^z((d(\psi_{ts})d(\psi_{t\star}))^2-1)}{(d(\psi_{ts})d(\psi_{t\star})e^z+1)(d(\psi_{ts})d(\psi_{t\star})+e^z)}-1.
\end{equation*}
Because $G_{sp}(z=0)=0$, if the first derivative
$G_{sp}^{(1)}(z=0)<0$, we will have $G_{sp}(z>0)<0$. Therefore,
\[G_{sp}^{(1)}(0)=2\times\sum_{t\in\Gamma_s\backslash
p}\frac{((d(\psi_{ts})d(\psi_{t\star}))^2-1)}{(d(\psi_{ts})d(\psi_{t\star})+1)(d(\psi_{ts})d(\psi_{t\star})+1)}-1<0\]\[\Rightarrow\sum_{t\in\Gamma_s\backslash
p}\frac{d(\psi_{ts})d(\psi_{t\star})-1}{d(\psi_{ts})d(\psi_{t\star})+1}<\frac{1}{2}.\]
\end{proof}
\textbf{Proof of Theorem \ref{theo:nonuniform convergence}}
\begin{proof}
Recall that in the proof of {\em uniform convergence condition},
we use an {\em error bound-variation function}
$G_{sp}(\log{\varepsilon})$, which is originally to describe
$(\log{\varepsilon_{sp}^{i+1}}-\log{\varepsilon_{ts}^{i}})$, for
$\forall (s,p)\in \mathbb{E}$. For each $T(\mathbb{G},v)$, given
$vu\in\mathbb{\bar{E}}$, let us introduce the following error
bound-variation function:
\begin{eqnarray*}
&G_{vu}(\{\log{\varepsilon_{w_iv}}\},\log{\varepsilon})=\sum_{w_i\in\Gamma_{v}\backslash u}\log{\frac{d(\psi_{l(w_i)v})^2\varepsilon_{w_iv}+1}{d(\psi_{l(w_i)v})^2+\varepsilon_{w_iv}}}-\log{\varepsilon},\\
&\log{\varepsilon_{w_iv}}=\sum_{w_j\in\Gamma_{w_i}\backslash
v}\log{\frac{d(\psi_{l(w_j)l(w_i)})^2\varepsilon_{w_jw_i}+1}{d(\psi_{l(w_j)l(w_i)})^2+\varepsilon_{w_jw_i}}},\\
&...\\
&\log{\varepsilon_{w_qw_p}}=\sum_{w_r\in\Gamma_{w_q}\backslash
w_p}\log{\frac{d(\psi_{l(w_r)l(w_q)})^2\varepsilon+1}{d(\psi_{l(w_r)l(w_q)})^2+\varepsilon}},
\end{eqnarray*}
where $\{w_r\}$ is the set of leaf nodes of $T(\mathbb{G},v)$.

To guarantee LBP to converge, it is sufficient to have
$G_{vu}(\log{\varepsilon})<0$ for $\forall \log{\varepsilon}>0$.
Because $G_{vu}(\log{\varepsilon}=0)=0$, when
$G_{vu}'(\log{\varepsilon}=0)<0$, we will definitely have
$G_{vu}(0<\log{\varepsilon}<\delta)<0$, where $\delta$ is a small
positive value. When $G_{vu}(\log{\varepsilon})$ is concave,
$\delta$ can be infinity so that the convergence of LBP is true
for $\forall \log{\varepsilon}>0$. However, because
$G_{vu}(\log{\varepsilon})$ is not guaranteed to be concave, we
will only obtain local convergence for an infinitesimal $\delta$.

Define
$f_{w_jw_i}(\varepsilon_{w_jw_i})=\log{\frac{d(\psi_{l(w_j)l(w_i)})^2\varepsilon_{w_jw_i}+1}{d(\psi_{l(w_j)l(w_i)})^2+\varepsilon_{w_jw_i}}}$.
Thus, we have the first derivative of\\
$G_{vu}(\{\log{\varepsilon_{w_iv}}\},\log{\varepsilon})$ as
follows:
\begin{equation*}
\frac{\partial
G_{vu}(\{\log{\varepsilon_{w_iv}}\},\log{\varepsilon})}{\partial
\log{\varepsilon}}=\sum_{w_i\in\Gamma_{v}\backslash
u}f_{w_iv}'\sum_{w_j\in\Gamma_{w_i}\backslash
v}f_{w_jw_i}'....\sum_{w_r\in\Gamma_{w_q}\backslash
w_p}f_{w_rw_q}'-1,
\end{equation*}
where $f'=\frac{\partial f(\log{\varepsilon})}{\partial
\log{\varepsilon}}=\frac{(d(\psi)^4-1)\varepsilon}{(d(\psi)^2\varepsilon+1)(d(\psi)^2+\varepsilon)}$.
Plugging $\log{\varepsilon}=0$ into the previous equation, we
obtain our non-uniform convergence condition.
\end{proof}
\textbf{Proof of Property \ref{prop:fixed or oscilloate}}
\begin{proof}
Let us analyze the fixed points by solving the set of equations
\begin{subequations}
\begin{eqnarray}
y=F(x)\label{eq:y=f(x)}\\
x=F(y)\label{eq:x=f(y)}
\end{eqnarray}
\end{subequations}
which corresponds to second order periodicity $x=F^2(x)$.  The set
of equations is depicted in Fig.~\ref{fig:binary LBP} for $a>b$
and $a<b$ respectively. We can easily find that $F(x)$ and $F(y)$
are symmetric with respect to $y=x$. Moreover, because $F(x)$ is
symmetric about the point $(0.5,0.5)$, we have $F(1-x)=1-F(x)$.
Therefore, it is easy to see that $F(x)$ and $F(y)$ are also
symmetric with respect to $y=1-x$. Let us check whether the two
functions are symmetric with respect to other lines such as
$y=\beta+\alpha x$. Substitute $y=\beta+\alpha x$ and
$x=\frac{1}{\alpha}(y-\beta)$ in~\eqref{eq:y=f(x)}.  We have
$\beta+\alpha
x=\frac{a(y-\beta)^k+b(\alpha-(y-\beta))^k}{(a+b)((y-\beta)^k+(\alpha-(y-\beta))^k)}$.
 For this equation to be always equivalent to~\eqref{eq:x=f(y)}, we have
$(\alpha=1,\beta=0)$ or $(\alpha=-1,\beta=1)$.  Thus, the set of
equations is only symmetric with respect to $y=x$ and $y=1-x$.

When $y=F(x)$ and $x=F(y)$ intersect, they must have crossing
points on $y=x$ or $y=1-x$. In the following, we will show that
they do not cross elsewhere. When $a>b$, let us assume these two
functions have one crossing point {\em A} not on $y=x$ and
$y=1-x$, which is illustrated in Fig.~\ref{fig:binary LBP} (a).
Due to the symmetry between $F(x)$ and $F(y)$, they must have the
other three crossing points $B,C$ and $D$ shown in
Fig.~\ref{fig:binary LBP} (a) respectively. Both functions must go
through those points. The first derivative of $F(x)$ is
$F^{(1)}(x)=\frac{k(a-b)x^{k-1}(1-x)^{k-1}}{(a+b)((1-x)^k+x^k)^2}=\left\{\begin{array}{c}
>0,a>b\\<0,a<b\end{array}\right.$, which shows that function $F(x)$ is
either monotonic increasing or monotonic decreasing.  Because
$y_B<y_A$, when $x_B>x_A$, we arrive at a contradiction with the
monotonic increasing property under the condition $a>b$. Similar
result is for $a<b$. According to Property \ref{prop:three fixed
points}, $y=F(x)$ and $x=F(y)$ have at most three real crossings
points with an arbitrary line. Therefore, we can see that the set
of equations will have at most three crossing points with either
$y=x$ or $y=1-x$.

The set of equations in~\eqref{eq:y=f(x)} and~\eqref{eq:x=f(y)}
has a naive fixed point $(0.5, 0.5)$. However, it is only stable
when the set of equations crosses nowhere else on $y=x$ and
$y=1-x$. When $a>b$ and
$F^{(1)}(\frac{1}{2})=\frac{k(a-b)}{(a+b)}>1$, we can see that the
belief network will either converge at fixed point E or at fixed
point F on $y=x$ in Fig.\ref{fig:binary LBP} (a). In this case,
the fixed point at $x=0.5$ is an unstable point. When $a<b$ and
$F^{(1)}(\frac{1}{2})<-1$, the belief network will eventually
oscillate between E and F on $y=1-x$, which is shown in
Fig.~\ref{fig:binary LBP} (b). The fixed point at $x=0.5$ is again
an unstable fixed point. Because $F(x)$ is symmetric with respect
to $(x=0.5,y=0.5)$, points E and F are symmetric with respect to
$(x=0.5,y=0.5)$.
\end{proof}
% Note: in this sample, the section number is hard-coded in. Following
% proper LaTeX conventions, it should properly be coded as a reference:

%In this appendix we prove the following theorem from
%Section~\ref{sec:textree-generalization}:

\vskip 0.2in
\bibliography{sample}

\end{document}